\documentclass[aps,pra,twocolumn]{revtex4-1}

\usepackage[T1]{fontenc}
\usepackage{amsmath,amsfonts,amssymb,amsthm,bbm,bm, braket}
\usepackage{graphicx}
\usepackage{wasysym}
\usepackage[usenames,dvipsnames]{color}
\usepackage[colorlinks,bookmarks=false,citecolor=NavyBlue,linkcolor=OliveGreen,url color=blue]{hyperref}
\newcommand{\be}{\begin{equation}}
\newcommand{\ee}{\end{equation}}
\newcommand{\1}{\mathbbm{1}}
\newtheorem{Prop}{Property}
\theoremstyle{break}        
\newtheorem{Theorem}{Theorem}
\newtheorem{Lemma}{Lemma}
\def\ua{\uparrow}
\def\da{\downarrow}

\newcommand{\CC}{\mathbb{C}}

\allowdisplaybreaks

\begin{document}

\title{Exact Spectral Form Factor in a Minimal Model of Many-Body Quantum Chaos}

\author{Bruno Bertini, Pavel Kos, and Toma\v z Prosen}

\affiliation{Department of Physics, Faculty of Mathematics and Physics, University of Ljubljana, Jadranska 19, SI-1000 Ljubljana, Slovenia}

\date{\today}

\begin{abstract}
The most general and versatile defining feature of quantum chaotic systems is that they possess an energy spectrum with correlations universally described by random matrix theory (RMT). This feature can be exhibited by systems with a well defined classical limit as well as by systems with no classical correspondence, such as locally interacting spins or fermions. Despite great phenomenological success, a general mechanism explaining the emergence of RMT without reference to semiclassical concepts is still missing. Here we provide the example of a quantum many-body system with no semiclassical limit ({\em no} large parameter) where the emergence of RMT spectral correlations is proven exactly. Specifically, we consider a periodically driven Ising model and write the Fourier transform of spectral density's two-point function, the spectral form factor, in terms of a partition function of a two-dimensional classical Ising model featuring a space-time duality. We show that the self-dual cases provide a {\em minimal model} of many-body quantum chaos, where the spectral form factor is demonstrated to match RMT for {\em all} values of the integer time variable $t$ in the thermodynamic limit. In particular, we rigorously prove RMT form factor for {\em odd} $t$, while we formulate a precise conjecture for {\em even} $t$. The results imply ergodicity for any finite amount of disorder in the longitudinal field, rigorously excluding the possibility of many-body localization. Our method provides a novel route for obtaining exact nonperturbative results in non-integrable systems.
\end{abstract}

\maketitle

The problem of finding a quantum analog of the classical concept of chaos has a long and fascinating history \cite{Gutzwiller,Haake,Maldacena}. For systems with chaotic and ergodic classical limit, the {\em quantum chaos} conjecture~\cite{CGV80,Berry81,BGS84} states that the statistical properties of energy spectrum are universal and given in terms of {\em random matrix theory} (RMT) \cite{Mehta}, where all matrix elements of the Hamiltonian are considered to be independent Gaussian random variables. An analogous result for chaotic maps, or periodically driven (Floquet) systems, relates the statistics of quasi-energy levels to circular ensembles of unitary random matrices \cite{Haake,Mehta}. This conjecture has been by now put on firm theoretical footing by clearly identifying contributions from periodic orbit theory and RMT for the simplest nontrivial measure of spectral correlations: the spectral form factor (SFF)
~\cite{Berry85,SR2001,S2002,Haake2004,Haake2005,Saito}. This, however, has been rigorously proven only for a specific type of single-particle models: the incommensurate quantum graphs~\cite{Weidenmuller1,Weidenmuller2}.

The situation is even less clear for non-integrable many-body systems with simple, say {\em clean and local}, interactions, where evidence of RMT spectral correlations is abundant \cite{ergodic1,ergodic2,ergodic3,ergodic4} but theoretical explanations are scarce. While for many-body systems of bosons with a large number of quanta per mode, or other models with small effective Planck's constant, a semiclassical reasoning may still be used 
\cite{Engl,Urbina,Guhr,Remy}, the intuition is completely lost and no methods have been known when it comes to fermionic or spin-$1/2$ systems. Very recently, a few steps of progress have been made. First, an analytic method analogous to the periodic orbit theory for spin-$1/2$ systems has been proposed in Ref.~\cite{KLP}. This method is able to establish RMT spectral fluctuations for long-ranged but non-mean-field non-integrable spin chains, however, it fails in the important extreme case of local interactions. Second, it has been shown in Refs.~\cite{Chalker,Chalker2} that Floquet local quantum circuits with Haar-random unitary gates have exact RMT SFF in the limit of large local Hilbert space dimension. Remarkably, in both cases the Thouless time, where universal RMT behavior sets in, scales as the logarithm of the system size~\cite{KLP,Chalker2} which is consistent with detailed numerical computations in Ref. \cite{Shenker}.

In this Letter we make a crucial step forward by providing the example of a locally interacting many-body system with finite local Hilbert space for which the SFF exactly approaches the RMT prediction in the thermodynamic limit (TL) at {\em all} times. Thus, we identify the first non-perturbative exactly solvable model displaying scale-free many-body quantum chaos.

More specifically, we consider the Floquet Ising spin-$1/2$ chain with transverse and longitudinal fields, described by the following Hamiltonian~\cite{KI_PRE, KI_JPA} 
\be
H_{\rm KI}[\boldsymbol h;t]=H_{\rm I}[\boldsymbol h]+\delta_p(t) H_{\rm K}\,.
\label{eq:ham}
\ee
Here $\delta_p(t)=\sum_{m=-\infty}^\infty\delta(t-m)$ is the periodic delta function and we defined 
\be
H_{\rm I}[\boldsymbol h]\equiv\! \sum_{j=1}^{L} \left\{J \sigma^{z}_j \sigma^z_{j+1}+ h_j \sigma^z_{j}\right\},\,\,\, H_{\rm K}\equiv b \sum_{j=1}^{L}  \sigma^x_j,
\label{eq:hamiltonians}
\ee
where we denote by $L$ the volume of the system, $\sigma_j^\alpha$, $\alpha\in\{x,y,z\}$, are the Pauli matrices at position $j$, and we impose $\sigma_{L+1}^{\alpha}=\sigma_1^{\alpha}$. The parameters $J, b$ are, respectively, the coupling of the Ising chain and the transverse kick strength, while $\boldsymbol h=(h_1,\ldots, h_L)$ describes a position dependent longitudinal field. Here and in the following, vectors of length $L$ are indicated by bold latin letters. For generic values of the longitudinal fields $\boldsymbol h$ the only symmetry possessed by the Hamiltonian \eqref{eq:ham} is time reversal. 

 \begin{figure}[t!]
\includegraphics[width=0.45\textwidth]{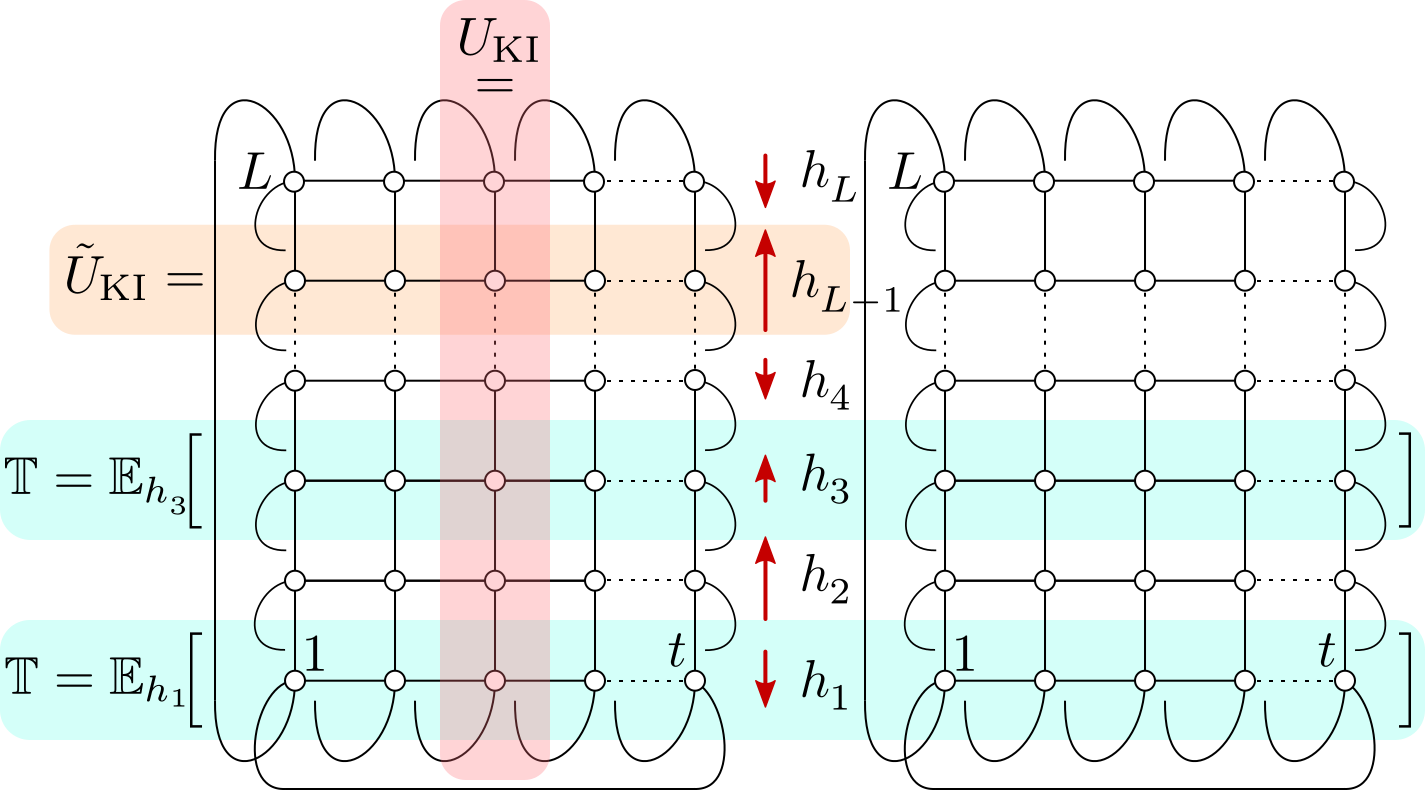}
\caption{Pictorial representation of $\bar K(t)$. The average over $h_j$ produces a transfer matrix $\mathbb T$ for all $j=1,\ldots,L$. Each column and row of the first lattice correspond respectively to the tranfer matrix $U_{\rm KI}[\boldsymbol h]$ and the dual transfer matrix $\tilde U_{\rm KI}[h_j \boldsymbol\epsilon]$. Each column and row of the second lattice correspond respectively to the complex conjugate transfer matrix $U_{\rm KI}[\boldsymbol h]^*$ and the complex conjugate dual transfer matrix $\tilde U_{\rm KI}[h_j \boldsymbol\epsilon]^*$.}
\label{fig:pictorial}
\end{figure}

The Floquet operator generated by \eqref{eq:ham} reads as 
\be
U_{\rm KI}[\boldsymbol h]= 
T\!\exp\!\!\left[-i\int_0^1\!\!\!\!{\rm d}s\, H_{\rm KI}[\boldsymbol h;s]\right]\!\!=
e^{-i H_{\rm K}}e^{-i H_{\rm I}[\boldsymbol h]}\,.
\label{eq:floquet}
\ee
In Floquet systems it is customary to introduce quasienergies $\{\varphi_n\}$ defined as the phases of the eigenvalues of the Floquet operator. The quasienergies take values in the interval 
$[0,2\pi]$ and their number is equal to the dimension of the Hilbert space ${\mathcal N=2^L}$. The quasienergy distribution function can then be written as ${\rho(\varepsilon)=\frac{2\pi}{\mathcal N}\sum_{n} \delta(\varepsilon- \varphi_n)}$. It is instructive to consider the connected two-point function of $\rho(\varepsilon)$, defined as~\cite{unfolding}
\be
{
r(\nu)=\frac{1}{2\pi}\int_{0}^{2\pi}\!\!\!\! {\rm d}\varepsilon\,\, \rho\left(\varepsilon+\frac{\nu}{2}\right)\rho\left(\varepsilon-\frac{\nu}{2}\right)-1\,.}
\ee
The Fourier transform of this quantity, known as the \emph{spectral form factor}, is the main object of our study 
\be
{
K(t)=\frac{{\mathcal N^2}}{2\pi}\!\!\int_{0}^{2\pi} \!\!\!\!\!\!{\rm d}\nu\, e^{i \nu t} r(\nu)=\!\!\sum_{m,n} e^{i(\varphi_m- \varphi_n)t}-{\mathcal N^2}\delta_{t,0}\,.}
\ee
This object can be efficiently calculated in the context of RMT. Since our system is time reversal invariant, the RMT prediction relevant to our case is that of the circular orthogonal ensemble, $K_{\rm COE}(t) = 2t - t \ln(1 + 2t/\mathcal N)$ for $0<t<\mathcal N$~\cite{Mehta}. SFF represents an extremely efficient and sensitive diagnostic tool for determining the spectral properties of a system. Any significant deviation from RMT is an indicator of non-ergodicity. For example, for integrable or localized systems, spectral fluctuations are conjectured to be Poissonian \cite{BerryTabor} and SFF is drastically different, $K(t) = {\cal N}$ for all $t>0$.

Floquet SFF is defined for integer times $t$ only (multiples of driving period) and for $t > 0$ admits a simple representation in terms of the Floquet operator \eqref{eq:floquet}
\be
K(t)= \left |{\rm tr}\left(U^t_{\rm KI}[\boldsymbol h]\right)\right|^2\,.
\label{eq:KtraceU}
\ee 
The trace of the Floquet operator can be thought of as the partition function of a two dimensional classical Ising model defined on a periodic rectangular lattice of size  $t\times L$
\begin{align}
\!{\rm tr}\left(U^t_{\rm KI}[\boldsymbol h]\right)&=\sum_{\{\boldsymbol{s}_\tau\}}\prod_{\tau=1}^{t} \braket{\boldsymbol{s}_{\tau+1}|e^{-i H_{\rm K}}e^{-i H_{\rm I}[\boldsymbol h]}|\boldsymbol{s}_{\tau}}\notag\\
&=[(\sin 2b)/(2i)]^{Lt/2} \sum_{\{s_{\tau,j}\}} e^{-i\mathcal E[\{s_{\tau,j}\},\boldsymbol h]}.
\label{eq:trUv1}
\end{align}
Here the configurations are specified by $\{\boldsymbol{s}_1,\ldots,\boldsymbol{s}_t\}\equiv \{s_{\tau,j}\}$, where $s_{\tau,j}\in\{\pm1\!\!\!\equiv\,\uparrow\downarrow\}$ for all $\tau,j$, and can be regarded as classical spin variables, $\ket{\boldsymbol s}$ is such that $\sigma^z_j\ket{\boldsymbol{s}}= s_j \ket{\boldsymbol{s}}$ and the energy of a configuration reads as  
\be
\mathcal E[\{s_{\tau,j}\},\boldsymbol{h}] = \sum_{\tau=1}^t  \sum_{j=1}^{L} (J s_{\tau,j} s_{\tau,j+1} + J'  s_{\tau,j} s_{\tau+1,j} + h_j s_{\tau,j})\label{eq:twodenergy}
\ee
where $J'= -\frac{\pi}{4}-\frac{i}{2}\log\tan b$.
Note that the Boltzmann weights of this model are generically complex.

Observing that \eqref{eq:twodenergy} couples only ``spins'' on neighbouring sites in both $t$ and $L$ directions, the partition function \eqref{eq:trUv1} can be written both as the trace of a transfer matrix propagating in the time direction and as the trace of a transfer matrix propagating in the space direction. This reveals the known duality transformation of the kicked Ising model~\cite{Guhr2, note2}. The transfer matrix in the time direction is clearly given by $U_{\rm KI}[\boldsymbol h]$, while the transfer matrix ``in space'', $\tilde{U}_{\rm KI}[\bold{h}_j]$, is given by the same algebraic form (\ref{eq:hamiltonians},\ref{eq:floquet}) exchanging $J$ and $J'$ but acting on a spin chain of $t$ sites. Moreover, it acts at non-stationary homogeneous field $\bold{h}_j=h_j\boldsymbol\epsilon$, where $\boldsymbol \epsilon=(1,\ldots,1)$ is a $t$-component constant vector. In other words, we have the identity  
\be
{\rm tr}\left(U^t_{\rm KI}[\boldsymbol h]\right) = {\rm tr} \left( \prod_{j=1}^L \tilde{U}_{\rm KI}[h_j \boldsymbol\epsilon]\right). 
\label{eq:dual}
\ee
Here $U_{\rm KI}[\boldsymbol h]$ acts on  ${\mathcal H_L= (\mathbb C^2)^{\otimes L}}$ and $\tilde{U}_{\rm KI}[h_j \boldsymbol\epsilon]$ acts on ${\mathcal H_t= (\mathbb C^2)^{\otimes t}}$. Note that $\tilde{U}_{\rm KI}[h_j \boldsymbol\epsilon]$ is generically non-unitary: it becomes unitary only for $|J|=|b|=\frac{\pi}{4}$ where $J'=\pm\frac{\pi}{4}$. We call these points of parameter space the ``self dual points'' and from now we focus on these.  

The SFF is known to be non-self averaging~\cite{Prange}. This means that $K(t)$ computed in a single system, \emph{i.e.} for fixed parameters $J,b,\boldsymbol h$, does not generically reproduce the ensemble average. In order to compare to RMT predictions we then need to average over an ensemble of similar systems. Here we consider a very natural form of averaging by introducing disorder (which we may switch off at the end of calculation): we assume that the longitudinal magnetic fields at different spatial points $h_j$ are independently distributed Gaussian variables with the mean value $\bar h$ and variance $\sigma^2 > 0$, and we average over their distribution. In other words, we consider  
\begin{align}
\bar K(t)\equiv\mathbb E_{\boldsymbol h}\left[{K(t)}\right]= \mathbb E_{\boldsymbol h}\left[{\rm tr}\left(U^t_{\rm KI}[\boldsymbol h]\right){\rm tr}\left(U^t_{\rm KI}[\boldsymbol h]\right)^*\right]\,,
\label{eq:averagedK}
\end{align}
where the symbol $\mathbb E_{\boldsymbol h}[\cdot]$ denotes the average over the longitudinal fields 
\be
\mathbb E_{\boldsymbol h}\left[ f(\boldsymbol h)\right] = \int_{-\infty}^\infty f(\boldsymbol h) \prod_{j=1}^L e^{-{(h_{j}-\bar h)^2}/{2\sigma^2}}\frac{{\rm d}h_j}{\sqrt{2\pi}\sigma}\,.
\label{eq:average}
\ee
The average in \eqref{eq:averagedK} mixes two copies of the classical Ising model \eqref{eq:twodenergy} with complex conjugate couplings.
After rewriting in terms of dual transfer matrices (\ref{eq:dual}), and noting $|{\rm tr}\,U|^2 = {\rm tr}\,(U\otimes U^*)$, we see that the average factorizes row-by-row, and local averaging results in translationally invariant coupling between two periodic rows of 
$t$ spins at the same spatial point. The resulting averaged SFF can again be interpreted as the trace of an appropriate transfer matrix in spatial direction (Fig.~\ref{fig:pictorial})
\be
\bar{K}(t) = {\rm tr}\left({\mathbb T}^L \right),
\label{eq:transfermatrix}
\ee
where the transfer matrix acts on $\mathcal H_t\otimes \mathcal H_t$ and reads as~\cite{note}
\be
{\mathbb T}\equiv \mathbb E_{h}\left[\tilde U_{\rm KI}[h {\boldsymbol\epsilon}]\otimes \tilde U_{\rm KI}[h {\boldsymbol\epsilon}]^*\right] =(\tilde U_{\rm KI}\otimes \tilde U_{\rm KI}^*) \cdot \mathbb O_\sigma\,.
\ee 
Here $\tilde U_{\rm KI} \equiv \tilde U_{\rm KI}[\bar h {\boldsymbol\epsilon}]$ and the local gaussian average is encoded in the following positive symmetric matrix
\be
\mathbb O_\sigma=\exp\left[{-\frac{1}{2}{\sigma^2}\left(M_z\otimes \1- \1 \otimes M_z\right)^2}\right]\,,
\label{eq:O}
\ee
where $M_\alpha\equiv\sum_{\tau=1}^t \sigma^\alpha_{\tau}$ for $\alpha\in\{x,y,z\}$. Note that, because of $\mathbb O_\sigma$, the matrix $\mathbb T$ is a non-unitary contraction.

The disorder averaged SFF $\bar K(t)$ can be computed numerically by evaluating \eqref{eq:KtraceU} for several values of the longitudinal fields and then taking the average \eqref{eq:average}. This can be done for small systems up to very large times, see Fig~\ref{fig:purenumerics}. Here, however, we follow a different route. Numerical data provide strong evidence for the validity of RMT for any fixed $t$, and any $\sigma>0$, in the TL $L\rightarrow\infty$. Indeed, the RMT prediction applies also for $t\ll\mathcal N$ when the system behaves as if it were effectively of infinite size. We then consider the TL and use Eq.~\eqref{eq:transfermatrix} to analytically compute $\bar K(t)$. This is done in two steps: (i) we map the seeming formidable problem of computing $\bar K(t)$ in our non-integrable many-body system into a simple problem in operator algebra; (ii) we solve the latter.

\begin{figure}[b!]
\includegraphics[width=0.48\textwidth]{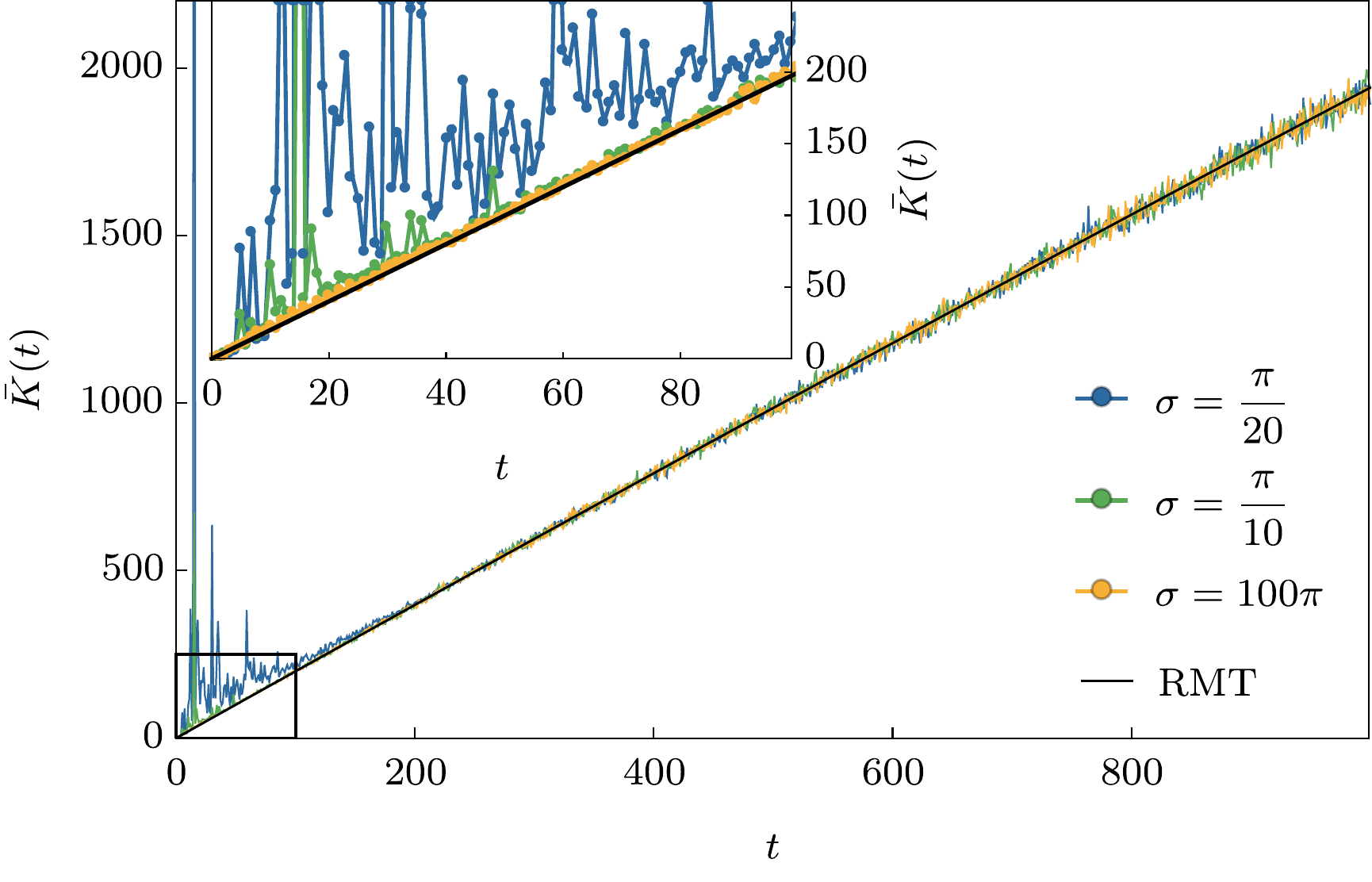}
\caption{SFF in the disordered kicked Ising model for ${J=b=\frac{\pi}{4}}$, ${L=15}$, and ${\bar h=0.6}$. The figure compares the time evolution of the SFF for different widths $\sigma$ of the disorder distribution. Inset: short-time window.  The large-time fluctuations are due to the finite number ($N=9490$) of disorder realizations.}
\label{fig:purenumerics}
\end{figure}

A numerical investigation indicates that, as long as ${\sigma\neq0}$, the spectral gap $\Delta=1-\max_{\lambda\in{\rm eigenvalues}(\mathbb T)}^{|\lambda|<1}|\lambda|$ remains finite for all mean fields and times, see Fig.~\ref{fig:gap}.
\begin{figure}[t!]
\includegraphics[width=0.46\textwidth]{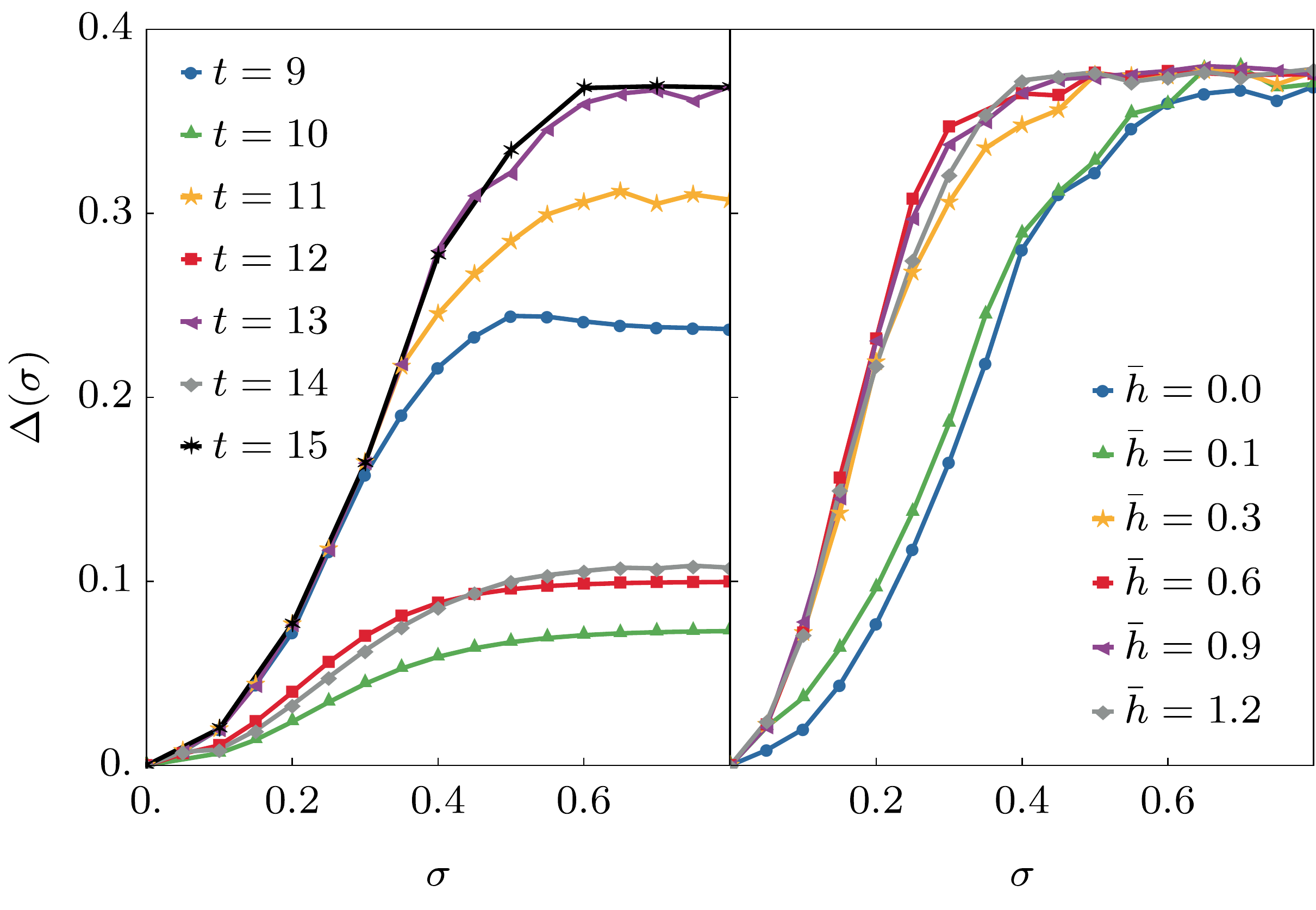}
\caption{
Spectral gap of $\mathbb T$ as a function of the disorder strength 
$\sigma$. The left panel shows $\Delta(\sigma)$ for $\bar h=0$ and different values of $t$: we observe a clear even-odd effect in the data but in both cases the gap approaches 
a finite limiting curve for large $t$. The right panel shows $\Delta(\sigma)$ for $t=13$ and 
different values of $\bar h$.
\label{fig:gap}}
\label{fig:gap}
\end{figure}
Therefore, the TL of the averaged SFF is entirely determined by the eigenvalues of $\mathbb T$ with largest magnitude. To find all such eigenvalues it is useful to exploit the following property~\cite{SM}, which is a consequence of the contractive nature of $\mathbb O_\sigma$ and of the unitarity of $\tilde U_{\rm KI}$. 
\begin{Prop}
\label{prop1}
(i) the eigenvalues of $\mathbb T$ have at most unit magnitude; (ii) even if $\mathbb T$ is generically not guaranteed to be diagonalisable, the algebraic and geometric multiplicites of any eigenvalue of magnitude $1$ coincide.
\end{Prop}
Let us then construct all eigenvectors $\ket{A}$ of $\mathbb T$ of unimodular eigenvalues. First, we note that all such ${\ket{A}}$ lie in the eigenspace of $\mathbb O_\sigma$ with unit eigenvalue. This is seen by expanding $\braket{A|\mathbb T^\dag\mathbb T|A}=1$ in an eigenbasis of $\mathbb O_\sigma$ 
\be
1=\braket{A|\mathbb T^\dag\mathbb T|A}=\braket{A|\mathbb O_\sigma^2|A} = \sum_n |\braket{A|n}|^2  o_{\sigma,n}^2\,,
\label{eq:shrink2}
\ee
where ${0 < {o_{\sigma,n}} \le 1}$ are the eigenvalues of $\mathbb O_\sigma$. Since $\ket{A}$ is normalized and $\{\ket{n}\}$ is complete, this is possible only if $\braket{A|n} =0$ for all ${o_{\sigma,n}<1}$.
In other words, $\ket{A}$ is a linear combination of eigenvectors of $\mathbb O_\sigma$ with unit eigenvalue, namely $\mathbb O_\sigma\!\ket{A} = \ket{A}$. Using the exponential form \eqref{eq:O} of $\mathbb O_\sigma$, we see that this condition means that all $\ket{A}$ are in the kernel of $M_z\otimes\1- \1\otimes M_z$. Putting all together, we have that the eigenvectors $\ket{A}$ associated to unimodular eigenvalues must satisfy 
\begin{align}
&\left(\tilde U_{\rm KI}\otimes \tilde U_{\rm KI}^*\right)\ket{A}=e^{i \phi}\ket{A}\,,\quad \phi\in[0,2\pi]\label{eq:cond1}\\
&\left(M_z\otimes\1- \1\otimes M_z\right)\ket{A}=0\,.\notag
\end{align}
These conditions can be turned into equations for operators over ${\cal H}_t$ as follows. Denoting by $\{\ket{ n}\}$ a basis of $\mathcal H_t$, we can expand a generic vector in $\mathcal H_t \otimes \mathcal H_t$ as 
\be
\ket{A}={\sum_{n,m}} A_{n,m} \ket{n}\otimes \ket{m}^*\,,
\ee    
where the $2^{2t}$ complex numbers $\{A_{n,m}\}$ are interpreted as the matrix elements of an operator $A$: $\braket{n|A|m}=A_{n,m}$. The operator $A$ is in one-to-one correspondence with the state $\ket{A}$ and we can rewrite the conditions \eqref{eq:cond1} as follows
\begin{align}
[A,M_z] =0\qquad \tilde U_{\rm KI}A \tilde U_{\rm KI}^\dag = e^{i \phi} A\,.
\label{eq:com}
\end{align}
After some simple manipulations~\cite{SM} we find
\begin{Prop}
\label{prop:newcond}
The relations \eqref{eq:com} are equivalent to 
\begin{align}
U A U^{\dag} = e^{i\phi} A\,,\quad[A,M_\alpha] =0\,,\quad \alpha\in\{x,y,z\}\,.\label{eq:newcond}
\end{align}
\end{Prop}
Here we defined the unitary operator
\be
U=\exp\left[i  \frac{\pi}{4} \sum_{\tau=1}^{t}\left( \sigma^z_\tau \sigma^z_{\tau+1}- \1\right)\right]\,.
\ee
The goal of step (i) is then achieved: the calculation of $\bar K(t)$ is reduced to finding all linearly independent matrices $A$ that fulfil Eq~\eqref{eq:newcond} for some $\phi$. Let us now consider the latter algebraic problem.

The property $U^2=\1$ implies ${\phi}\in\{ 0,\pi\}$. Namely, the unimodular eigenvalues of $\mathbb T$  are either $1$ or $-1$. By exact numerical diagonalization of $\mathbb T$ we find that the eigenvalues $-1$ are much rarer than $+1$ and are observed only for small systems, see Tab.~\ref{tab:Teigenvalues}. In particular, for odd $t$ we have the following additional symplification~\cite{SM} 
\begin{Prop}
$\phi=0$ for odd $t$.
\end{Prop}

For odd $t$  we then need to determine all linearly independent matrices $A$ commuting with the set ${\cal M}=\{U,M_x,M_y,M_z\}$. A subset of all possible operators commuting with ${\cal M}$ is found by considering the common symmetries: reflection $R$ and one-site shift $\Pi$ on a periodic chain of $t$ sites
\be
\Pi=\prod_{\tau=1}^{t-1}P_{\tau, \tau+1}\qquad R=\prod_{\tau=1}^{\lfloor{t/2}\rfloor}P_{\tau, t+1-\tau}\,.
\label{eq:symmetries}
\ee
Here $P_{\tau, \omega}=\frac{1}{2}\1+\frac{1}{2}\sum_{\alpha}\sigma^\alpha_\tau\sigma^\alpha_\omega$ is the elementary permutation operator (transposition). These operators generate the so called \emph{dihedral group} (see, \emph{e.g.}, \cite{Serre}) 
\be
\mathcal G_t = \left\{\Pi^n R^m;\;\; n\in\{0,\ldots,t-1\},\; m\in\{0,1\}\right\}\,,
\ee
which is the symmetry group of a polygon with $t$ vertices. All elements of $\mathcal G_t$ commute with ${\cal M}$ and we have~\cite{SM} 
\begin{Prop}
\label{prop:numberofindependentoperators}
The number of linearly independent elements of $\mathcal G_t$ is $2t$ for $t\geq6$, $2t-1$ for $t\in\{1,3,4,5\}$, and $2$ for $t=2$.
\end{Prop}
We thus have a lower bound on the number of independent matrices $A$ fulfilling \eqref{eq:newcond} and hence on the value of the averaged SFF for odd $t$. Our main result is to show that such lower bound is also an upper bound, namely 
\begin{Theorem}
\label{prop:basicconjecture}
For odd $t$, any $A$ simultaneously commuting with all elements of  $\{U,M_x,M_y,M_z\}$, is of the form
\be
A=\sum_{n=0}^{t-1}\sum_{m=0}^1 a_{n,m} \Pi^n R^m \,,\qquad  a_{n,m}\in\mathbb C\,.
\label{eq:generalcommutingoperator}
\ee
\end{Theorem}

\noindent
See \cite{SM} for a proof. As the number of such linearly independent $A$ is the multiplicity of eigenvalue $1$ of $\mathbb T$, and since there is a finite gap between unit circle and the rest of the spectrum, we have 
\be
\lim_{L\rightarrow\infty}\bar K(t)   =\begin{cases}
 2t-1\,, &t\leq5\\
 2t\,, &t\geq 7
 \end{cases}\,,\quad \text{$t$ odd}\,.
\label{eq:Kodd}
\ee

For even $t$ the situation is more complicated. In this case, we identify an additional independent operator  besides ${\cal G}_t$ spanning the commutant of ${\cal M}$. This operator can be written as a projector $\ket{\psi}\!\!\bra{\psi}$, where we introduced a $t-$spin singlet state
\be
\ket{\psi}=\frac{1}{2^t} \prod_{\tau =1}^{t/2} \left(1-P_{\tau,\tau+t/2}\right)\ket{\underbrace{\da,\ldots,\da}_{t/2},\underbrace{\ua,\ldots,\ua}_{t/2}}\,,
\label{eq:psi1}
\ee
satisfying ${U\ket{\psi}=-\ket{\psi}}$, ${M_{x,y,z}\ket{\psi}=0}$, $\Pi \ket{\psi}=-\ket{\psi}$, $R\ket{\psi}=(-1)^{t/2} \ket{\psi}$. Moreover, for $t\in\{8,10\}$ we identify the second additional operator commuting with the set $\cal M$~\cite{SM}. Finally, for $t\in\{6,10\}$ we construct two operators satisfying \eqref{eq:newcond} with eigenphase $\phi=\pi$ \cite{SM}. All these additional operators, except (\ref{eq:psi1}), appear to be a short-time fluke and are observed only for $t$ smaller than 11. We are then lead to conjecture 
\be
\lim_{L\rightarrow\infty}\bar K(t) = 2t+1\,,\quad t>11\,,\quad \text{$t$ even}\,.
\label{eq:Keven}
\ee
This conjecture, together with the exact result \eqref{eq:Kodd}, is in agreement with exact diagonalization of $\mathbb T$ on chains of length $t \leq 17$, see Tab.~\ref{tab:Teigenvalues}.
\begin{table}
\begin{tabular}{|| c || c | c | c | c | c | c | c | c | c | c | c | c | c | c | c | c ||}
\hline
$t$  & 2 & 3 & 4 & 5 & 6 & 7 & 8 & 9 & 10 & 11 & 12 & 13 & 14 & 15 & 16 & 17\\ 
\hline
$\#_{+1}$ & 2 & 5 & 7 & 9 & 13 & 14 & 18 & 18 & 22 & 22 & 25 & 26 & 29 & 30 & 33 & 34\\
\hline
 $\#_{-1}$ & 0 & 0 & 0 & 0 & 2 & 0 & 0 & 0 & 2 & 0 & 0 & 0 & 0 & 0 & 0 & 0\\
\hline
\end{tabular}
\caption{Number of eigenvalues $1$ and $-1$ of the transfer matrix $\mathbb T$ determined via exact diagonalization for $t\le 17$.}
\label{tab:Teigenvalues}
\end{table}

The results \eqref{eq:Kodd} and \eqref{eq:Keven} are remarkable: we fully recovered 2-point RMT spectral fluctuations (in the TL) in a simple non-integrable spin-1/2 chain with local interactions. A key step of our calculation was to average over the distribution of independent longitudinal fields $\boldsymbol h$. 
This average introduces a finite gap in the spectrum of the transfer matrix $\mathbb T$  and selects the $2t$ ``universal'' eigenvalues out of the exponentialy many eigenvalues of $\mathbb T$. Note that any nonvanishing $\sigma$ is sufficient for this astonishing simplification to occur. Moreover, after the TL is taken there is no additional dependence of the result on the disorder variance $\sigma^2$, we can then consider the limit $\sigma\rightarrow0$ corresponding to a clean system. Finally, our result does not depend on the particular distribution of the longitudinal fields, as long as they are independent and identically distributed random variables; a different choice modifies the form of \eqref{eq:O} but not the TL result. Since our analysis is carried out in the TL, it is unable to access RMT physics at time-scales growing with $L$, such as level repulsion emerging at $t\sim 2^L$.

Our proof of ergodicity pertains to some special, self-dual, points in the parameter space of the system. At these points the system is ``maximally ergodic'' as the Thouless time does not grow with $L$. We have numerically verified the stability of the ergodic behaviour under perturbations around the self-dual points. In this case, however, the Thouless time becomes an increasing function of $L$, as expected in generic chaotic systems.

A striking consequence of our result is a rigorous proof of non-existence of many-body localization \cite{Huse,Papic,Roderich} at any self-dual point in our model ($J,b\in \{\pm \frac{\pi}{4}\}$) for any amount of 
uncorrelated disorder in the longitudinal field. Indeed, knowing that $K(t) = 2t$ for {\em odd} $t \ge 7$ is enough to exclude localization which should be connected to Poissonian behavior.

The technique developed gives a new way of analytically treating non-integrable systems and suggests immediate applications in several directions. First, one can apply it to compute the bipartite entanglement entropy dynamics starting from a random separable state, testing recent conjectures~\cite{RandomCircuitsEnt, Nahum, Keyserlingk} on its universal linear behaviour in ergodic systems. Moreover, our method can be used to rigorously approach ETH by studying averages and higher moments of distributions of expectation values of local observables. Finally, one can use our technique to evaluate dynamical correlation functions of local observables.  A preliminary analysis shows that at the self-dual point in the TL they vanish for all $t \ge 1$, consistently with an $L$-independent Thouless time.

\emph{Acknowledgments.} The work has been supported by ERC Advanced grant 694544 -- OMNES and the grant P1-0044 of Slovenian Research Agency. 
We acknowledge fruitful discussions with M. Akila, B. Gutkin, M. Ljubotina, A. Nahum, T. H. Seligman, and P. \v Semrl.

\onecolumngrid

\pagebreak

\setcounter{equation}{0}
\setcounter{figure}{0}
\setcounter{table}{0}
\setcounter{Prop}{0}
\setcounter{Theorem}{0}

\begin{center}
{\large{\bf Appendices}}
\end{center}

\begin{center}
{\large{\bf Supplemental Material for\\
 ``Exact Spectral Form Factor in a Minimal Model of Many-Body Quantum Chaos''}}
\end{center}

Here we report the proofs of the properties and the theorem described in the main text together with some useful complementary information. In particular 
\begin{itemize}
\item[-] Section \ref{app:proofprops} contains the proofs of Properties 1--4;
\item[-] Section \ref{app:proof} contains the proof of Theorem 1;
\item[-] Section \ref{app:evencase} contains some results for even $t$;
\item[-] Section \ref{app:integrablecase} contains the diagonalization of the Floquet operator, Eq. (3) of the main text, in the integrable case $\boldsymbol h=0$;
\item[-] Section \ref{app:numerics} contains some details on the numerical methods adopted to produce Tab.~I, Fig. 2, and Fig. 3 of the main text; 
\end{itemize}

\section{Proofs of Properties 1--4}
\label{app:proofprops}

Here we report the proofs of Properties 1--4 from the main text. For convenience of the reader, we also precisely state each property before proving it. 

\begin{Prop}
The following facts hold
\begin{itemize}
\item[(i)] the eigenvalues of $\mathbb T$ (\emph{cf}. Eq.~(13) of the main text) have at most unit magnitude; 
\item[(ii)] even if $\mathbb T$ is generically not guaranteed to be diagonalisable, the algebraic and geometric multiplicites of any eigenvalue of magnitude $1$ coincide;
\end{itemize} 
\end{Prop}
\begin{proof}
Let $\ket{w}$ be an eigenvector associated with an eigenvalue with largest magnitude, then we have   
\be
\braket{w|\mathbb O_\sigma^2|w}=\braket{w|\mathbb T^\dag\mathbb T|w}=|\lambda|^2\,.
\ee 
Expanding in the eigenbasis of $\mathbb O_\sigma$ we have 
\be
\braket{w|\mathbb O_\sigma^2|w} = \sum_n |\braket{w|n}|^2\,o_{\sigma,n}^2 \leq \sum_n |\braket{w|n}|^2   = 1\,,
\label{eq:shrink}
\ee
where we used that $\ket{w}$ is normalised and the eigenvalues $o_{\sigma,n}$ of $\mathbb O_{\sigma}$ are in the interval $[0,1]$. This proves the point (i). For the point (ii) we first note that the matrix $\mathbb T$ is not normal 
\be
\mathbb T ^\dag\mathbb T =\mathbb O_\sigma^2\,,\qquad 
\mathbb T\mathbb T^\dag =(\tilde U_{\rm KI}\otimes \tilde U_{\rm KI}^*)  \mathbb O_\sigma^2 (\tilde U_{\rm KI}\otimes \tilde U_{\rm KI}^*) ^\dag\,.
\ee
There is then no general theorem ensuring that $\mathbb T$ is diagonalizable. To prove that, nonetheless, the algebraic and geometric multiplicites of any eigenvalue of magnitude $1$ coincide we use \emph{reductio ad absurdum}: let us assume that there is an eigenvalue $\lambda$ with unit magnitude corresponding to a non-trivial Jordan block 
and see that this leads to a contradiction. Let $\ket{w}$ be the eigenvector associated to $\lambda$, since the Jordan block is non-trivial we have a vector $\ket{v}$ such that  
\be
\mathbb T \ket{v} = \lambda \ket{v}+ \alpha \ket{w}\,,\qquad \alpha \neq 0\,,
\ee
and 
\be
\braket{v|w}=0\,.
\ee
This implies (noting that $\lambda^* \lambda = 1$)
\be
\braket{v|\mathbb O_\sigma^2|v}=\braket{v|\mathbb T^\dag \mathbb T | v} = \braket{v|v}+|\alpha|^2\braket{w|w} + \lambda\alpha^* \braket{w|v} + \lambda^*\alpha \braket{v|w} =1+|\alpha|^2>1\,.
\label{eq:absurd}
\ee
This result is in contradiction with \eqref{eq:shrink} concluding the proof of point (ii). 
\end{proof}

\begin{Prop}
The relations (16) of the main text imply 
\begin{align}
[A,M_\alpha] &=0\,, & &\alpha\in\{x,y,z\}\,,\label{eq:newcond1}\\
U A U^{\dag} &= e^{i \phi} A\,, & &\phi\in[0,2\pi]\,,\label{eq:newcond2}
\end{align}
where we defined 
\be
U=\exp\left[i  \frac{\pi}{4} \sum_{j=1}^{t}\left(\sigma^z_j \sigma^z_{j+1}- \1\right)\right]\,.
\ee
\end{Prop}
\begin{proof}
Let us consider the operator $\tilde U_{\rm KI} M_z \tilde U_{\rm KI}^\dag$. Using Equations (16) of the main text it is immediate to see that this operator commutes with $A$
\be
A \tilde U_{\rm KI} M_z \tilde U^\dag_{\rm KI} =  e^{-i \phi} \tilde U_{\rm KI} A M_z \tilde U^\dag_{\rm KI} = e^{-i \phi} \tilde U_{\rm KI} M_z A \tilde U^\dag_{\rm KI} = \tilde U_{\rm KI} M_z \tilde U^\dag_{\rm KI} A\,.
\ee
Using the explicit form of $\tilde U_{\rm KI}$ (setting $|J|=|b|=\pi/4$) we have 
\be
\tilde U_{\rm KI} M_z \tilde U_{\rm KI}^\dag = e^{\mp i \frac{\pi}{4} M_x} M_z e^{\pm i \frac{\pi}{4} M_x} = \sum_{j=1}^t  e^{\pm i \frac{\pi}{2} \sigma_j^x} \sigma_j^ze^{\mp i \frac{\pi}{2} \sigma_j^x}  = \mp M_y\,.
\ee
So $M_y$ commutes with $A$. Since both $M_z$ and $M_y$ commute with $A$ we have that 
\be
M_x = e^{i \frac{\pi}{4} M_z} M_y e^{-i \frac{\pi}{4} M_z}\,,
\ee
also commutes with $A$: this proves \eqref{eq:newcond1}. The relation \eqref{eq:newcond2} follows from \eqref{eq:newcond1} and the first one of Eqs.~(16) of the main text. 
\end{proof}

\begin{Prop}
The case $\phi=\pi$ does not give solutions to \eqref{eq:newcond1}--\eqref{eq:newcond2} for odd $t$. 
\end{Prop}
\begin{proof}
To prove this property we consider the self-dual integrable (transverse field) kicked Ising Floquet operator 
\be
U_{\rm TFKI}\equiv e^{-i  \frac{\pi}{4} \sum_{j=1}^{t}\sigma^z_j \sigma^z_{j+1}}e^{+i\frac{\pi}{4}  M_x }\,.
\ee
The operator $U_{\rm TFKI}$ $\phi$-commutes with $A$ as a consequence of Property~2 and can be explicitly diagonalized: its diagonalization is reviewed in Section~\ref{app:integrablecase}. Writing the $\phi$-commutation in an eigenbasis of $U_{\rm TFKI}$ we have 
\be
\braket{{n,\rm r}|A|{m,\rm r'}} (e^{i (e_{n,\rm r}-e_{m,\rm r'})}-e^{i \phi})=0\,,
\label{eq:Ainteigenvalues}
\ee
where we denoted by $\ket{n,\rm r}$ the eigenstates of $U_{\rm TFKI}$. The index $\rm r\in\{\rm R, \rm NS\}$ distinguishes the two spin-flip symmetry sectors (\emph{cf}. Section~\ref{app:integrablecase}). Since $A$ is not the zero matrix, Eq.~\eqref{eq:Ainteigenvalues} implies that there exist some $n$ and $m$ such that 
\be
e^{i (e_{n,\rm r}-e_{m,\rm r'})}-e^{i \phi}=0\,.
\ee
Using the explicit expression for the quasi-energy differences in the integrable kicked Ising model at the self dual points, derived in Section~\ref{app:integrablecasediff}, we have
\be
e_{n,\rm r}-e_{m,\rm r'} = \frac{2\pi}{t} N_{n,m}+\frac{\pi}{4 t} (\eta_{\rm r}-\eta_{\rm r'})=\phi\,, \qquad\qquad N_{n,m}\in\mathbb Z\,,\qquad\qquad\eta_{\rm R}=-\eta_{\rm NS}=-1\,,
\ee
implying 
\be
{8}N_{n,m}+(\eta_{\rm r}-\eta_{\rm r'})=\frac{4 t \phi}{\pi} \qquad\qquad N_{n,m}\in\mathbb Z\,,\qquad\qquad\eta_{\rm R}=-\eta_{\rm NS}=-1\,.
\ee 
For $\phi=\pi$ the equation is solved for $t$ even by $\rm r = \rm r'$ and $N_{n,m}=t/2$. There is, however, no solution for $\phi=\pi$ and $t$ odd. 
\end{proof}

Before considering the proof of Property 4, it is useful to introduce orthogonal projectors $\{Y_k\}_{k=0,\ldots,t-1}$ on the fixed momentum eigenspaces $\{V_k\}_{k=0,\ldots,t-1}$, 
$\bigoplus_{k=0}^{t-1}V_k = {\cal H}_t$. Such operators are defined as follows 
\be
Y_k=\frac{1}{t}\sum_{j=0}^{t-1}e^{\frac{2\pi i}{t} j k } \Pi^j\,,\qquad\qquad k\in\{0,\dots,t-1\}\,,
\ee
and satisfy the orthogonality and the projection property
\be
Y_k Y_p = \delta_{k-p}  Y_k\,.
\ee
Moreover, they fulfil
\be
R Y_k R = Y_{t-k}\,.
\ee
We also introduce an additional set of operators 
\be
Y'_k=R Y^{\phantom{\prime}}_k\,,\qquad\qquad k\in\{0,\dots,t-1\}\,.
\ee
The operators $\{Y_k,Y'_k\}_{k=0,\ldots,t-1}$ form a closed multiplicative algebra, defined by 
\be
Y^{\phantom{\prime}}_k Y^{\phantom{\prime}}_p = \delta_{k-p}  Y^{\phantom{\prime}}_k\,,\qquad Y'_k Y^{\phantom{\prime}}_p = \delta_{k-p} Y'_k,\qquad Y^{\phantom{\prime}}_k Y'_p = \delta_{k+p} Y'_p, \qquad Y'_k Y'_p = \delta_{k+p} Y^{\phantom{\prime}}_p\,.
\ee
The linear mapping between $\{Y^{\phantom{\prime}}_k,Y'_k\}_{k=0,\ldots,t-1}$ and the elements of the dihedral group $\mathcal G_t$ is explicitly inverted as follows 
\be
\sum_{k=0}^{t-1} Y_k \,e^{-\frac{2\pi i}{t} j k} = \Pi^j\,,\qquad\qquad \sum_{k=0}^{t-1} Y'_k \,e^{-\frac{2\pi i}{t} j k} = R \Pi^j\,.
\ee
Since the mapping is invertible, we can restate the Property 4 as follows
\begin{Prop}
\label{prop:restatedprop4}
The number of linearly independent operators in the set $\{Y^{\phantom{\prime}}_k,Y'_k\}_{k=0,\ldots,t-1}$ is $2t$ for $t\geq6$, $2t-1$ for $t\in\{1,3,4,5\}$, and $2$ for $t=2$.
\end{Prop}
\begin{proof}
To prove this statement it is useful to distinguish the cases of $t$ even and $t$ odd. For odd $t$, all $\{Y^{\phantom{\prime}}_k,Y'_k\}_{k\neq0}$ are linearly independent. This is seen by writing these operators in the basis of momentum eigenstates and noting that each operator is non-zero only within a block of states with a given total momentum. Distinct operators are non zero on distinct, non-overlapping, blocks: this proves their linear independence. Noting that all $\{Y^{\phantom{\prime}}_k,Y'_k\}_{k\neq0}$ are zero when reduced to the zero-momentum block, we have that the only two operators which can be linearly dependent are $Y^{\phantom{\prime}}_0$ and $Y_0'$. In Lemma~\ref{Lemma1} of Section~\ref{app:proof} we prove that for $t\leq5$ they are linearly dependent, namely 
\be
Y^{\phantom{\prime}}_0+Y_0'=0\,,
\ee 
while they are independent for $t\geq 6$. In the case of even $t$ there is an additional special pair of operators, $Y_{t/2}$ and $Y_{t/2}'$, acting on the same momentum block (where all other $Y_k$ and $Y'_k$ are zero). In Lemma~\ref{Lemma2} of Section~\ref{app:evencase} we prove that $\{ Y^{\phantom{\prime}}_{t/2}, Y'_{t/2}\}$ are linearly independent for all even $t\ge 4$, while for $t=2$, $Y^{\phantom{\prime}}_{t/2}+Y'_{t/2}=0$. Putting all together we found that the set $\{Y^{\phantom{\prime}}_k,Y'_k\}_{k=0,\ldots,t-1}$ is composed of $2t-1$ linearly independent operators for $t\in\{1,3,4,5\}$ and by $2t$ linearly independent operators for $t\geq6$. In the special case $t=2$ we have just two linearly independent operators: $Y_0,Y_{t/2}=Y_1$. This concludes the proof. 
\end{proof}

\section{Proof of Theorem 1}
\label{app:proof}

Here we present the proof of Theorem 1. We start by restating it in terms of the operators $\{Y^{\phantom{\prime}}_k,Y'_k\}_{k=0,\ldots,t-1}$, introduced in the previous section 
\begin{Theorem}
\label{prop:basicconjecture}
For odd $t$, any $A$ simultaneously commuting with all the elements of $\{U,M_x,M_y,M_z\}$, is necessarily of the form
\be
A=\sum_{k=0}^{t-1} c_{k}Y_k+\sum_{k=0}^{t-1} d_{k}Y'_k\,,\qquad\qquad\qquad  c_{k},d_k\in\mathbb C\,.
\label{eq:generalcommutingoperator}
\ee
\end{Theorem}
To prove it, we reformulate it in the following equivalent way   
\begin{Theorem}
\label{prop:basicconjecture2}
Considering the representation of the algebra $\mathcal K$ generated by $\{U,M_x,M_y,M_z\}$ in a subspace $V_k\subset {\cal H}_t$ of fixed momentum $k$, the following facts hold  for odd $t$
\begin{itemize}
\item[(i)] If $k\neq0$ the representation is irreducible.
\item[(ii)] If $k=0$ the representation is irreducible for $t\leq5$ and is composed of two inequivalent irreducible representations for $t>5$. 
\item[(iii)] The representations in the two fixed momentum subspaces $V_k$ and $V_p$ are equivalent for $p= k$ or $p=t-k$ and inequivalent otherwise. 
\end{itemize}
\end{Theorem}
Let us first show that Theorem \ref{prop:basicconjecture2} is equivalent to Theorem~\ref{prop:basicconjecture}. 
\begin{proof}
Theorem \ref{prop:basicconjecture2} implies Theorem~\ref{prop:basicconjecture} as a consequence of Schur's Lemma. This is easily seen as follows. Let $A$ be a matrix commuting with the algebra $\mathcal K$ generated by $\{U,M_x,M_y,M_z\}$. Then we have 
\be
E_k A_{k,p}= A_{k,p} E_p\,,\qquad\forall E\in \mathcal K\,,
\ee
where
\be
A_{k,p}\equiv Y_k A Y_p\,,\qquad\qquad E_k\equiv Y_k E = E Y_k\,.
\ee
The facts (i), (ii), and (iii), combined with Schur's Lemma imply  
\begin{align}
&A_{k,p}=\delta_{k,p}\, c_k\, Y_{k} + \delta_{p,t-k}\, d_k\, Y'_k\,, & &k=1,\ldots,t\,,\qquad\qquad c_{k},d_k\in\mathbb C\,,\\
&A_{0,0}= c'_{0}\, Y_{0, a} + d'_{0}\, Y_{0, b}\,,  & &c'_{0},d'_0\in\mathbb C\,.
\end{align}
Here we used that $Y_k$ acts as the identity operator in the subspace $V_k$, while $Y'_k$ acts as the isomorphism between the subspaces $V_k$ and $V_{t-k}$. Moreover, we denoted by $Y_{0,a}$ and $Y_{0,b}$ the projectors on the subspaces $V_{0,a},V_{0,b}\subset V_0$ such that 
\be
V_0=V_{0,a} \oplus V_{0,b}\,.
\ee
The spaces $V_{0,a}$ and $V_{0,b}$ carry two inequivalent irreducible representations of $\mathcal K$ (for $t\leq5$ one of the two is trivial). Writing $Y_{0, a},Y_{0, b}$ in terms of $Y_{0}$ and $Y'_{0}$ we have
\be 
A_{0,0}= c'_{0}\, Y_{0, a} + d'_{0}\, Y_{0, b}= c_{0}\, Y_{0} + d_{0}\, Y'_{0}\,,
\ee
for appropriate $c_{0}, d_{0}\in\mathbb C$. Combining all together, we have 
\be
A=\sum_{k,p=0}^{t-1} A_{k,p}=\sum_{k=0}^{t-1} c_{k}Y_k+\sum_{k=0}^{t-1} d_{k}Y'_k\,.
\ee
This proves Theorem~\ref{prop:basicconjecture}. 

The reverse direction is proven by \emph{reductio ad absurdum}. If the representation of $\mathcal K$ in a fixed momentum sector were not satisfying (i), (ii) or (iii),  we could find operators commuting with the algebra $\mathcal K$ which are not of the form \eqref{eq:generalcommutingoperator}.  
\end{proof}

Let us now prove Theorem~\ref{prop:basicconjecture2}. Our strategy is to construct a basis of the momentum eigenspace $V_k$. We show that
\begin{itemize}
\item[(1)]  For $k\neq0$ all basis vectors are mapped into one another by elements of the algebra $\mathcal K$. 
\item[(2)] For $k=0$ all reflection symmetric vectors are mapped into one another by elements of the algebra $\mathcal K$. The same holds for reflection antisymmetric vectors. Moreover, we show that for $t\leq5$ the reflection anti-symmetric subspace is the zero space.  
\item[(3)] There is no invertible matrix $C$ such that  
\be
(U)_p = C (U)_k C^{-1}\,,\qquad\qquad (M_j)_p = C (M_j)_k C^{-1},\qquad\qquad\qquad \text{for}\qquad k\not\in\{ p, t-p\}\,,
\label{eq:point3}
\ee
where by $(M)_p$ we denote the matrix $M$ restricted to $V_p$. 
\end{itemize}
It is immediate to see that (1) $\Rightarrow$ (i), (2) $\Rightarrow$ (ii), and (3) $\Rightarrow$ (iii).

{We construct the basis of the sector with fixed momentum $k$ as follows. Let us define the states  
\be
{\ket{\bullet \underbrace{\circ\cdots\circ}_{\ell_1-1}\bullet\underbrace{\circ\cdots\circ}_{\ell_2-1}\bullet\cdots\bullet\underbrace{\circ\cdots\circ}_{\ell_a-1}\bullet}}^{(k)}=\frac{1}{\sqrt{t}}\sum_{j=0}^{t-1} e^{\frac{2\pi i}{t} k j}\Pi^j \sigma^-_{1}\sigma^-_{\ell_1+1}\sigma^-_{\ell_1+\ell_2+1}\cdots\sigma^-_{m}\ket{\underbrace{\ua\ldots\ua}_t}
\label{eq:states}
\ee
where we introduced the ``length of the state'' $m$ such that 
\be
m=1+\sum_{j=1}^a\ell_{j}<t\,.
\ee
These states are represented by a string of $m$ symbols starting and ending with $\bullet$, while the $m- 2$ symbols in the ``bulk'' can be either ``bullets'' $\bullet$ or ``holes'' $\circ$. This means that for every length $m\geq2$ there are $2^{m-2}$ states, whereas two special cases are 
\begin{align}
&m=1 \qquad\longrightarrow \qquad\ket{\bullet}^{(k)}\,,\\
&m=2 \qquad\longrightarrow \qquad\ket{\bullet\bullet}^{(k)}\,.
\end{align}  
We then generically represent a state \eqref{eq:states} as 
\be
\ket{\bullet \,{\boldsymbol\sigma}^s_{m-2}\bullet}^{(k)}
\label{eq:pictorialrep}
\ee
where ${\boldsymbol\sigma}^s_{m-2}$ is a generic string of $\circ$ and $\bullet$ of length $m-2$ and used the superscript $s$ to indicate that the string ${\boldsymbol\sigma}_{m-2}$ has exactly $s$ holes. Note that the states \eqref{eq:pictorialrep} have momentum $k$
\be
\Pi \ket{\bullet \,{\boldsymbol\sigma}^s_{m-2}\bullet}^{(k)} = e^{-i k} \ket{\bullet \,{\boldsymbol\sigma}^s_{m-2}\bullet}^{(k)}\,. 
\ee
Moreover they are eigenstates of $M_z$ and $U$
\be
M_z\ket{\bullet \,{\boldsymbol\sigma}^s_{m-2}\bullet}^{(k)}= (t-2m + 2s)\ket{\bullet \,{\boldsymbol\sigma}^s_{m-2}\bullet}^{(k)}\,, \qquad\qquad\qquad U\ket{\bullet \,{\boldsymbol\sigma}^s_{m-2}\bullet}^{(k)} =(-1)^{\nu}\ket{\bullet \,{\boldsymbol\sigma}^s_{m-2}\bullet}^{(k)} \,,
\ee
where $\nu$ is the number of ``disconnected islands'' of bullets in the string $\bullet \,{\boldsymbol\sigma}^s_{m-2}\bullet$. In the representation \eqref{eq:states} the latter is given by one plus the number of $\ell_1,\ell_2,\ldots,\ell_a$ larger than $1$ for $m<t$ and by the number of of $\ell_1,\ell_2,\ldots,\ell_a$ larger than $1$ for $m=t$. Note that, for $k\neq0$, the state $\ket{\bullet}^{(k)}$ is the largest highest weight state among those of the irreducible representations of $SU(2)$ living in $V_k$. This state is the one with maximal eigenvalue of $M_z$ and is therefore unique.

For $k\neq0$ the states \eqref{eq:pictorialrep} form a complete set in $V_k$ while for $k=0$ the set is complete if we add the  ferromagnetic states, specifically the ``vacuum state''  and the ``all bullets'' state 
\be
\ket{\emptyset}=\ket{\underbrace{\ua\cdots\ua}_t}\,, \qquad\qquad \ket{\underbrace{\bullet\cdots\bullet}_t}= \ket{\underbrace{\da\cdots\da}_t}\,,
\ee
which are translationally invariant and thus appear only in $V_0$. The states \eqref{eq:pictorialrep} are, however, not all linearly independent. While for $m < (t+1)/2$ the states are orthonormal, some of the states with $m \geq (t+1)/2$ can be represented by a string with shorter length or they have multiple representations with the same length. We then construct a basis $\mathcal B_k$ of $V_k$ as follows 
\begin{align}
\mathcal B_k =\left\{\ket{\bullet}^{(k)},\ket{\bullet \bullet}^{(k)}\right\} &\cup \left\{\ket{\bullet \,{\boldsymbol\sigma}_{m-2}\bullet}^{(k)}: \forall {\boldsymbol\sigma}_{m-2}\,,\quad m\in \{3,\ldots,({t-1})/{2}\}\right\}\notag\\
&\cup \left\{\ket{\bullet \,{\boldsymbol\sigma}_{m-2}\bullet}^{(k)}: \forall {\boldsymbol\sigma}_{m-2}\,,\quad m\in \{({t+1})/{2},\ldots,t-1\}\right\}'
\label{eq:basis}
\end{align}
where we denote by  $\{\cdots\}'$ the maximal set of linearly independent vectors. Note in particular that for each $m$ the so-called ``$m$-block state''
\be
\ket{\underbrace{\bullet\bullet\cdots \bullet}_m}^{(k)}
\label{eq:blockstates}
\ee
is always included in the basis because it is impossible to represent it with lower $m$. Finally, the sector  $k=0$
 is special because it is invariant under the reflection symmetry $R$. It is then convenient to consider a basis of simultaneous eigenvectors of  $M_z$, $U$, $\Pi$ and $R$. This is easily done by taking 
\be
\bar {\mathcal B}_0 = S_+ {\mathcal B}_0 \cup S_-  {\mathcal B}_0 \cup \left\{\ket{\emptyset},\ket{{\bullet\cdots\bullet}}\right\}\,,
\ee
where ${\mathcal B}_0$ is the basis \eqref{eq:basis} for $k=0$, and 
\be
S_\pm=\frac{1}{2}\left(\1 \pm R\right)\,,
\label{eq:Sproj}
\ee
projectors to even/odd reflection subspaces.
Before starting with the proof it is useful to introduce the following operators which generate the algebra ${\cal K}$
\be
M^{\pm}= \frac{1}{2}\left(M_x \pm i M_y\right)\,,\qquad\qquad\qquad P_{\pm} = \frac{1}{2}\left(\1 \pm U\right)\,.
\ee  
Three of these operators have a very simple action on the states \eqref{eq:pictorialrep}. Specifically,
\be
P_{\pm}\ket{\bullet \,{\boldsymbol\sigma}^s_{m-2}\bullet}^{(k)} = \frac{1}{2}\left(1\pm(-1)^\nu\right) \ket{\bullet \,{\boldsymbol\sigma}^s_{m-2}\bullet}^{(k)}\,,
\label{def:Ppm}
\ee
where $\nu$ is, again, the number of disconnected islands of bullets in $\bullet \,{\boldsymbol\sigma}^s_{m-2}\bullet$, and
\be
M^+\ket{\bullet \,{\boldsymbol\sigma}^s_{m-2}\bullet}^{(k)} = \sum'_{h}  \ket{\bullet \,{\boldsymbol\sigma}^s_{m-2,h}\bullet}^{(k)}\, + \text{states of length smaller than $m$}\,,
\label{def:M+}
\ee
where the sum is over the positions of the bullets in ${\boldsymbol\sigma}^s_{m-2}$ and ${\boldsymbol\sigma}^s_{m-2,h}$ is the string obtained by ${\boldsymbol\sigma}^s_{m-2}$ by removing the bullet in position $h$. However, the action of $M^-$ on the basis state is simple only if it is
combined with the action of an appropriate projector $P_\sigma$, where $\sigma=(-1)^\nu$ 
\be
P_\sigma M^-\ket{\bullet \,{\boldsymbol\sigma}^s_{m-2}\bullet}^{(k)}  = 
e^{-\frac{2\pi i k}{t}} \ket{\bullet\!\bullet{\boldsymbol\sigma}^s_{m-2}\bullet}^{(k)}
+\ket{\bullet \,{\boldsymbol\sigma}^s_{m-2}\bullet\!\bullet}^{(k)} 
+
\sum''_{b}  \ket{\bullet \,{\boldsymbol\sigma}^s_{m-2,b}\bullet}^{(k)}.
\label{def:M-}
\ee
The sum is now over all the positions of the holes in ${\boldsymbol\sigma}^s_{m-2}$, and ${\boldsymbol\sigma}^s_{m-2,b}$ is the string obtained from ${\boldsymbol\sigma}^s_{m-2}$ by putting a bullet in a vacant position $b$, such that the total number of islands of holes in the string
${\boldsymbol\sigma}^s_{m-2,b}$ is the same as in ${\boldsymbol\sigma}^s_{m-2}$ (i.e., we can only add a bullet to any of contiguous clusters of bullets, without changing the number of islands of bullets). A few additional useful operators are given by the following Lemma

\begin{Lemma}
\label{lemma:additionaloperators}
The following operators belong to the algebra $\mathcal K$
\begin{align}
&M^{++}\equiv\sum_{\tau=1}^{t} \sigma_\tau^+\sigma_{\tau+1}^+\,,& &M^{--}\equiv\sum_{\tau=1}^{t} \sigma_\tau^-\sigma_{\tau+1}^-\,,\label{eq:goaloperators1}\\
&M^{\overbrace{z\!+\!\ldots\!+\!z}^k}\equiv\sum_{\tau=1}^{t} \sigma_\tau^z\sigma_{\tau+1}^+\cdots\sigma_{\tau+k}^+\sigma_{\tau+k-1}^z\,,& &M^{\overbrace{z\!-\!\ldots\!-\!z}^k}\equiv\sum_{\tau=1}^{t} \sigma_\tau^z\sigma_{\tau+1}^-\cdots\sigma_{\tau+k}^-\sigma_{\tau+k-1}^z\,,
\label{eq:goaloperators}\\
&M^{\overbrace{+\cdots+}^{t}} \equiv \prod_{\tau=1}^t \sigma^+_\tau\,,& &M^{\overbrace{-\cdots-}^{t}}\equiv \prod_{\tau=1}^t \sigma^-_\tau\,,\label{eq:goaloperators3}
\end{align}
where $k\leq t$ is odd and we defined $\sigma^\pm_\tau=\frac{1}{2}(\sigma^x_\tau\pm i \sigma^y_\tau)$. 
\end{Lemma}
\begin{proof}
To prove the statement we explicitly construct these operators using generators of $\mathcal K$.

Let us consider 
\be
U M_x U= \sum_{\tau=1}^{t} U \sigma_\tau^x U = \sum_{\tau=1}^{t} \sigma_\tau^x e^{i \frac{\pi}{2} (\sigma^z_{\tau-1}\sigma^z_{\tau}+\sigma^z_{\tau}\sigma^z_{\tau+1})} = - \sum_{\tau=1}^{t} \sigma_{\tau-1}^z\sigma_\tau^x \sigma_{\tau+1}^z\,.
\ee
Rotating by $\pi/4$ around $x$ axis we then find 
\be
e^{i \frac{\pi}{4} M_x}U M_x U e^{-i \frac{\pi}{4} M_x} = - \sum_{\tau=1}^{t} \sigma_{\tau-1}^y\sigma_\tau^x \sigma_{\tau+1}^y\,.
\ee
Iterating the same operations $(k-1)/2$ times, for $k\leq t$ odd, we find 
\be
\underbrace{e^{i \frac{\pi}{4} M_x}U\cdots e^{i \frac{\pi}{4} M_x}U}_{(k-1)/2} M_x \underbrace{U e^{-i \frac{\pi}{4} M_x}\cdots U e^{-i \frac{\pi}{4} M_x}}_{(k-1)/2}= - \sum_{\tau=1}^{t} \sigma_{\tau}^y\sigma_{\tau+1}^x \cdots \sigma_{\tau+t-2}^x \sigma_{\tau+k-1}^y\,.
\label{eq:opk}
\ee
In particular, for $k=t$ we have 
\be
\underbrace{e^{i \frac{\pi}{4} M_x}U\cdots e^{i \frac{\pi}{4} M_x}U}_{(t-1)/2} M_x \underbrace{U e^{-i \frac{\pi}{4} M_x}\cdots U e^{-i \frac{\pi}{4} M_x}}_{(t-1)/2}= - \sum_{\tau=1}^{t} \sigma_{\tau}^y\sigma_{\tau+1}^x \cdots \sigma_{\tau+t-2}^x \sigma_{\tau+t-1}^y\,.
\ee
Multiplying by $e^{i \frac{\pi}{2} M_x}$ we find 
\be
e^{i \frac{\pi}{2} M_x}\underbrace{e^{i \frac{\pi}{4} M_x}U\cdots e^{i \frac{\pi}{4} M_x}U}_{(t-1)/2} M_x \underbrace{U e^{-i \frac{\pi}{4} M_x}\cdots U e^{-i \frac{\pi}{4} M_x}}_{(t-1)/2}= - i^{(t-2)} \sum_{\tau=1}^{t} \sigma_{\tau}^z\sigma_{\tau+1}^z\,.
\ee
The operators \eqref{eq:goaloperators1} are then obtained as follows. First we construct  
\begin{align}
&\sum_{\tau=1}^{t} \sigma_{\tau}^x\sigma_{\tau+1}^x=e^{i \frac{\pi}{4} M^y}\sum_{\tau=1}^{t} \sigma_{\tau}^z\sigma_{\tau+1}^ze^{-i \frac{\pi}{4} M^y}\,,\qquad \sum_{\tau=1}^{t} \sigma_{\tau}^y\sigma_{\tau+1}^y=e^{i \frac{\pi}{4} M^x}\sum_{\tau=1}^{t} \sigma_{\tau}^z\sigma_{\tau+1}^ze^{-i \frac{\pi}{4} M^x}\,,\\
&\sum_{\tau=1}^{t} (\sigma_{\tau}^x\sigma_{\tau+1}^y+\sigma_{\tau}^y\sigma_{\tau+1}^x)=\frac{i}{2} \left[\sum_{\tau=1}^{t} \sigma_{\tau}^x\sigma_{\tau+1}^x, M^z\right]\,.
\end{align}
Then we take the following linear combinations
\be
\sum_{\tau=1}^{t} \sigma_{\tau}^x\sigma_{\tau+1}^x-\sum_{\tau=1}^{t} \sigma_{\tau}^y\sigma_{\tau+1}^y\pm i\sum_{\tau=1}^{t} (\sigma_{\tau}^x\sigma_{\tau+1}^y+ \sigma_{\tau}^y\sigma_{\tau+1}^x)=M^{\pm\pm}\,.
\ee
The operators \eqref{eq:goaloperators} are constructed starting from $\sum_{\tau=1}^{t} \sigma_{\tau}^y\sigma_{\tau+1}^x \cdots \sigma_{\tau+k}^x \sigma_{\tau+k-1}^y$ (\emph{cf}. \eqref{eq:opk}) as follows. First we construct  
\be
\sum_{\tau=1}^{t} \sigma_{\tau}^z\sigma_{\tau+1}^x \cdots \sigma_{\tau+k}^x \sigma_{\tau+k-1}^z=e^{i \frac{\pi}{4} M^x}\sum_{\tau=1}^{t} \sigma_{\tau}^y\sigma_{\tau+1}^x \cdots \sigma_{\tau+k}^x \sigma_{\tau+k-1}^y e^{-i \frac{\pi}{4} M^x}\,.
\ee
Then we take subsequent commutators with $M^z$
\be
C_{-1}\equiv0\,,\qquad
C_0\equiv\sum_{\tau=1}^{t}  \sigma_{\tau}^z\sigma_{\tau+1}^x \cdots \sigma_{\tau+k}^x \sigma_{\tau+k-1}^z\,, \qquad 
C_n\equiv\frac{1}{n}\left(\left(\frac{i}{2}\right) \left[C_{n-1},M^z\right]+(k+2-n) C_{n-2}\right)\,.
\ee
The last step is to take the following linear combinations
\be
\sum_{j=0}^{k-2}(\pm i)^j C_j  = M^{\overbrace{z\!\pm\!\ldots\!\pm\!z}^{k}}\,.
\ee
Finally, the operators \eqref{eq:goaloperators3}  are obtained by starting from 
\be
\prod_{\tau=1}^t \sigma^x_\tau = e^{i \frac{\pi}{2} (M^x-1)} \,,
\ee
taking subsequent commutators with $M^z$ 
\begin{align}
& D_{-1}\equiv0\,,\qquad
D_0\equiv\prod_{\tau=1}^t \sigma^x_\tau \,, \qquad 
D_n\equiv\frac{1}{n}\left(\left(\frac{i}{2}\right) \left[D_{n-1},M^z\right]+(k+2-n) D_{n-2}\right)\,,
\end{align}
and constructing the following linear combinations 
\be
\sum_{j=0}^{t}(\pm i)^j D_j  = M^{\overbrace{\pm\!\ldots\!\pm}^{t}}\,.
\ee
This concludes the proof.  
\end{proof}

\subsubsection{Proof of (1)}
We are now in a position to prove (1): we consider $k\neq0$ and show that we can map the states of the basis $\mathcal B_k$ into one another using elements of the algebra. Namely, for each pair $\ket{\bullet{\boldsymbol\sigma}_{m-2}\bullet}^{(k)}$ and $\ket{\bullet{\boldsymbol\sigma}'_{m'-2}\bullet}^{(k)}$ of states in $\mathcal B_k$ we construct an operator  $A_{{\boldsymbol\sigma},{\boldsymbol\sigma}'}\in \mathcal K$ such that
\be
{}^{(k)}\braket{\bullet{\boldsymbol\tau}_{m-2}\bullet|A_{{\boldsymbol\sigma},{\boldsymbol\sigma}'}|\bullet{\boldsymbol\tau}'_{m'-2}\bullet}^{(k)} = \delta_{{\boldsymbol\tau},{\boldsymbol\sigma}}\delta_{{\boldsymbol\tau'},{\boldsymbol\sigma'}}\,,\qquad \qquad \forall  \ket{\bullet{\boldsymbol\tau}_{m-2}\bullet}^{(k)},\ket{\bullet{\boldsymbol\tau}'_{m'-2}\bullet}^{(k)}\in\mathcal B_k\,.
\label{eq:Aconn}
\ee 
To construct such operator, it is useful to introduce the operator $B_{\boldsymbol\sigma}$, mapping the largest highest weight state $\ket{\bullet}^{(k)}$ to the basis vector $\ket{\bullet{\boldsymbol\sigma}_{m-2}\bullet}^{(k)}$. Namely
\be
B_{\boldsymbol\sigma} \ket{\bullet}^{(k)}= \ket{\bullet{\boldsymbol\sigma}_{m-2}\bullet}^{(k)}\,.
\ee
In terms of $B_{\boldsymbol\sigma}$, an operator $A_{{\boldsymbol\sigma},{\boldsymbol\sigma}'}$ fulfilling \eqref{eq:Aconn} can be written as 
\be
A_{{\boldsymbol\sigma},{\boldsymbol\sigma}'} = B_{\boldsymbol\sigma} \ket{\bullet}^{(k)(k)}\!\!\bra{{\bullet}} B_{\boldsymbol\sigma}^\dag\,,
\ee
where ${(\cdot)}^\dag$ represents Hermitian conjugation. Our goal is to show that such $A_{{\boldsymbol\sigma},{\boldsymbol\sigma}'}$ is in $\mathcal K$. It is immediate to see that the projector $\ket{\bullet}^{(k)(k)}\!\!\bra{{\bullet}}$ is in the algebra $\mathcal K$. Indeed, it can be written as  
\be
\ket{\bullet}^{(k)}\!\!\!\!\!\!\!\!{\phantom{\ket{\bullet}}}^{(k)}\!\bra{{\bullet}} =  \frac{1}{4 \cos^2\left(\frac{2\pi k}{t}\right)}M^{\overbrace{z+\!\ldots\!+z}^{t}}M^{\overbrace{z-\!\ldots\!-z}^{t}}\,.
\ee 
where $M^{\overbrace{z\pm\!\ldots\!\pm z}^{t}}$ are constructed in Lemma~\ref{lemma:additionaloperators} and the operator on the right hand side should be interpreted as an element of the representation of the algebra ${\cal K}$ in the fixed momentum subspace $V_k$. Moreover, if the operator $B_{\boldsymbol\sigma}$ is in the algebra so is its Hermitian conjugate. This is seen by noting  
\be
U^\dag=U\,,\qquad M_x^\dag =M_x\,, \qquad M_y^\dag =M_y\,, \qquad M_z^\dag=M_z\,,
\ee
implying that the Hermitian conjugate of any sum of products of $\{U,M_\alpha\}_{\alpha=x,y,z}$ is again a sum of products of $\{U,M_\alpha\}_{\alpha=x,y,z}$. 

The problem is then reduced to proving that the operator $B_{\boldsymbol\sigma}$ is in the algebra, or, in other words, that starting from $\ket{\bullet}^{(k)}$ we can access any state in $\mathcal B_k$ using elements of $\mathcal K$.

In order to proceed and avoid the problem of linear dependence in our coding of states for $m \ge (t+1)/2$, we make use of the following trick.
Consider spin chain of {\em double size} $2t$ and define analogous basis states 
\begin{eqnarray}
&&\ket{\overline{\bullet{\boldsymbol\sigma}_{m-2}\bullet}}^{(k)} \equiv \ket{\bullet{\boldsymbol\sigma}_{m-2}\bullet}_{2t}^{(2k)},\qquad \text{or}\\
&&{\ket{\overline{\bullet \underbrace{\circ\cdots\circ}_{\ell_1-1}\bullet\underbrace{\circ\cdots\circ}_{\ell_2-1}\bullet\cdots\bullet\underbrace{\circ\cdots\circ}_{\ell_a-1}\bullet}}}^{(k)}=\frac{1}{\sqrt{2t}}\sum_{j=0}^{2t-1} e^{\frac{2\pi i}{t} k j} \bar{\Pi}^j \sigma^-_{1}\sigma^-_{\ell_1+1}\sigma^-_{\ell_1+\ell_2+1}\cdots\sigma^-_{m}\ket{\underbrace{\ua\ldots\ua}_{2t}}
\label{eq:statesb}
\end{eqnarray}
spanning $V^{(2k)}_{2t}$ as defined by Eq. (\ref{eq:states}) for chain length $2t$ and momentum $2k$. Here, $\bar{\Pi}$ is a cyclic shift on $(\CC^2)^{\otimes 2t}$.
Now, let us define the map $Q: V^{(2k)}_{2t} \to V^{(k)}_t$ completely specified by its action on the basis states
\be
Q \ket{\overline{\bullet{\boldsymbol\sigma}_{m-2}\bullet}}^{(k)} = 
\begin{cases}
 \ket{\bullet{\boldsymbol\sigma}_{m-2}\bullet}_{t}^{(k)}, & m < t; \\
e^{\frac{2\pi i (n+1+h)k}{t}} \ket{ {\boldsymbol\sigma}''_{l}\!\bullet \underbrace{\circ\cdots \circ}_{2t-m}  \bullet {\boldsymbol\sigma}'_{n} }^{(k)}, & h > t;\\
 0, & \text{otherwise}.
\end{cases}
\label{eq:Q}
\ee
Here we wrote 
\be
{\boldsymbol\sigma}_{m-2} = \bullet {\boldsymbol\sigma}'_{n} \underbrace{\circ\cdots \circ}_h {\boldsymbol\sigma}''_{l}\bullet,
\ee 
where $h$ is the largest number of contiguous holes in ${\boldsymbol\sigma}_{m-2}$ and $n+h+l=m-2$. Writing the algebra generators represented on the Hilbert space of spin chains of length $2t$ as $\bar{M}_\alpha$ and $\bar{P}_\pm=\frac{1}{2}(\1\pm\bar{U})$, it is easy to check that
\be
Q \bar{M}_\alpha = M_\alpha Q\,,\qquad Q \bar{P}_\pm = P_\pm Q\,.
\ee
We shall write the corresponding linear operator algebra generated by $\bar{P}_\pm,\bar{M}_\pm$ as $\bar{\cal K}$. For example,  $\bar{M}^{\pm\pm}$ now defined as $
Q\bar{M}^{\pm\pm}=M^{\pm\pm}Q$, then read consistently
\be
\bar{M}^{\pm\pm} = \sum_{\tau=1}^{2t} \sigma^\pm_\tau \sigma^\pm_{\tau+1},\quad{\text{with}}\quad 2t+1\equiv 1\,.
\ee

Below we shall show that $\ket{\overline{\bullet}}^{(k)}$ can be connected to any $\ket{\overline{\bullet{\boldsymbol\sigma}_{m-2}\bullet}}^{(k)}$, for $m < t$, by the action of some element from $\bar{\cal K}$. It then follows trivially that $\ket{\bullet}^{(k)}$ connects to $\ket{\bullet{\boldsymbol\sigma}_{m-2}\bullet}^{(k)}$ by the corresponding element of ${\cal K}$ (just replacing generators
$\bar{M}^\pm,\bar{P}_\pm$ by $M^\pm,P_\pm$ in the generator representation of an element of $\bar{\cal K}$ or ${\cal K}$).

We start by noting 
\be
\bar P_- \bar M^+ \ket{\overline{\underbrace{\bullet\bullet\cdots \bullet}_m}}^{(k)}= (1+e^{\frac{2\pi i}{t} k})\ket{\overline{\underbrace{\bullet\bullet\cdots \bullet}_{m-1}}}^{(k)}\qquad\qquad \bar P_- \bar M^- \ket{\overline{\underbrace{\bullet\bullet\cdots \bullet}_m}}^{(k)}= (1+e^{-\frac{2\pi i}{t} k})\ket{\overline{\underbrace{\bullet\bullet\cdots \bullet}_{m+1}}}^{(k)}\,.
\label{eq:mapamongblocks}
\ee
where the right hand sides are never zero for $t$ odd. Moreover
\be
\bar P_+ \bar M^{+}\ket{\overline{\bullet\bullet\bullet}}^{(k)} = \ket{\overline{\bullet\circ\bullet}}^{(k)}\,.
\label{eq:m3connection}
\ee
From these relations it follows that we can map $\ket{\overline\bullet}^{(k)}$ into any state of length $m\leq3$ and to any block state \eqref{eq:blockstates}.

We now proceed using an inductive argument. Assuming that we can access every state $\ket{\overline{\bullet \,{\boldsymbol\sigma}^s_{n-2}\bullet}}^{(k)}$ of length $n<m$ and every state $\ket{\overline{\bullet \,{\boldsymbol\sigma}_{m-2}^{r}\bullet}}^{(k)}$ of length $m$ and $r<s$ holes we shall prove that we can access every state of length $m\geq4$ and $s\leq m-2$ holes
\be
\ket{\overline{\bullet \,{\boldsymbol\sigma}^s_{m-2}\bullet}}^{(k)}\,.
\label{eq:goalstates}
\ee
We first show that the unique state with $s=m-2$ holes is directly obtained from the block state of length $m$ as follows. First we note that  
\be
(\bar P_- \bar M^+)^{m-3} \bar P_+ \bar M^+ \ket{\overline{\underbrace{\bullet\bullet\cdots \bullet}_m}}^{(k)} = \ket{\overline{\bullet\underbrace{\circ\cdots \circ}_{m-2} \bullet}}^{(k)} + A \ket{\overline{\bullet}}^{(k)}\,.
\ee
Here $A \ket{\overline \bullet}^{(k)}$ represents a combination of states of length strictly smaller than $m$, which by inductive assumption can be obtained from $\ket{\overline \bullet}^{(k)}$ by applying an appropriate element of the algebra, $A\in \bar{\mathcal K}$. Since $ \ket{\overline{\underbrace{\bullet\bullet\cdots \bullet}_m}}^{(k)}$ is connected to $ \ket{\overline \bullet}^{(k)}$ by elements of the algebra we have 
\be
A \ket{\overline \bullet}^{(k)} = B  \ket{\overline{\underbrace{\bullet\bullet\cdots \bullet}_m}}^{(k)}\qquad \text{for some}\qquad B\in \bar{\mathcal K}\,.
\ee
Then, defining 
\be
\left((\bar P_- \bar M^+)^{m-3}\bar P_+ \bar M^+\right)_{\rm red}  = (\bar P_- \bar M^+)^{m-3}\bar P_+ \bar M^+  - B \qquad\qquad \left((\bar P_- \bar M^+)^{m-3}\bar P_+ \bar M^+\right)_{\rm red}\in \bar{\mathcal K}\,,
\ee
we have 
\be
\left((\bar P_- \bar M^+)^{m-3}\bar P_+ \bar M^+\right)_{\rm red} \ket{\overline {\underbrace{\bullet\bullet\cdots \bullet}_m}}^{(k)} = \ket{\bullet\underbrace{\circ\cdots \circ}_{m-2} \bullet}^{(k)} \,.
\ee
Consequently, in the following we can restrict to $s<m-2$.  

Considering any state $\ket{\overline{\bullet \,{\boldsymbol\sigma}_{m-3}\bullet}}^{(k)}$ of length $m-1$ we have 
\be
\bar P_\sigma \bar M^- \ket{\overline{\bullet \,{\boldsymbol\sigma}_{m-3}\bullet}}^{(k)} = e^{-\frac{2\pi i k}{t}} \ket{\overline{\bullet \bullet \,{\boldsymbol\sigma}_{m-3}\bullet}}^{(k)}+ \ket{ \overline{\bullet \,{\boldsymbol\sigma}_{m-3}\bullet \bullet}}^{(k)} + \sum''_p\ket{\overline{\bullet \,{\boldsymbol\sigma}_{m-3,p}\bullet}}^{(k)}  \,,
\ee 
where $\sigma$ is chosen equal to $+$ or $-$ in order to maintain the number of ``disconnected islands'' of bullets when adding the new one and the last term on the r.h.s. is a sum of states obtained from  $\ket{\overline{\bullet \,{\boldsymbol\sigma}_{m-3}\bullet}}^{(k)}$ by adding a bullet in position $1<p<m-1$. The allowed values of $p$ are the positions of the holes ${\boldsymbol\sigma}_{m-3}$ restricted by $\bar P_\sigma$. By the inductive assumption, the states $\ket{\overline{\bullet \,{\boldsymbol\sigma}_{m-3,p}\bullet}}^{(k)}$ can be accessed from $\ket{\overline \bullet}^{(k)}$ with an appropriate element of the algebra $\bar{\mathcal K}$. Since $\ket{\overline{\bullet \,{\boldsymbol\sigma}_{m-3}\bullet}}^{(k)}$ is connected to $\ket{\overline\bullet}^{(k)}$ by elements of the algebra we can represent the last term on the r.h.s. as follows 
\be
\sum_p  \ket{\overline{\bullet \,{\boldsymbol\sigma}_{m-3,p}\bullet}}^{(k)}={C}\ket{\overline{\bullet \,{\boldsymbol\sigma}_{m-3}\bullet}}^{(k)}\,,
\ee
for an appropriate $C\in\bar{\mathcal K}$. Defining then 
\be
(\bar P_\sigma \bar M^-)_{\rm red}\equiv \bar P_\sigma \bar M^- -C\,,\qquad \qquad (\bar P_\sigma \bar M^-)_{\rm red}\in \bar{\mathcal K}\,,
\ee
we have 
\be
(\bar P_\sigma \bar M^-)_{\rm red} \ket{\overline{\bullet \,{\boldsymbol\sigma}_{m-3}\bullet}}^{(k)} = e^{-\frac{2\pi i k}{t}} \ket{\overline {\bullet \bullet \,{\boldsymbol\sigma}_{m-3}\bullet}}^{(k)}+ \ket{ \overline{\bullet \,{\boldsymbol\sigma}_{m-3}\bullet \bullet}}^{(k)}\,.
\label{eq:indcondgeneric}
\ee 
Note that this relation alone is not sufficient to fix all the states of length $m$ given those of length $m-1$. It gives $2^{m-3}$ independent conditions while the total number of states is $2^{m-2}$. 

Restricting \eqref{eq:indcondgeneric} to states with $s$ holes we have 
\be
(\bar P_\sigma \bar M^-)_{\rm red} \ket{\overline{\bullet \,{\boldsymbol\sigma}^{s}_{m-3}\bullet}}^{(k)} = e^{-\frac{2\pi i k}{t}} \ket{\overline{\bullet \bullet \,{\boldsymbol\sigma}^{s}_{m-3}\bullet}}^{(k)}+ \ket{ \overline{\bullet \,{\boldsymbol\sigma}^{s}_{m-3}\bullet \bullet}}^{(k)}\,.
\label{eq:indcondsholes1}
\ee

An additional condition is found by using more of our inductive data: the states of length $m$ and $s-1$ holes. Applying $\bar M^+$ we find   
\be
(\bar M^+)_{\rm red} \ket{\overline{\bullet \bullet \,{\boldsymbol\sigma}_{m-3}^{s-1}\bullet}}^{(k)} =\ket{\overline{\bullet \circ \,{\boldsymbol\sigma}_{m-3}^{s-1}\bullet}}^{(k)} +  \sum'_{p}\ket{\overline{\bullet \bullet \,{\boldsymbol\sigma}^{s}_{m-3,p}\bullet}}^{(k)}\,,
\label{eq:indcondsholes2p1}
\ee
where $(\bar M^+)_{\rm red}\in\bar{\mathcal K}$ is defined as $\bar M^+-B$ for an appropriate $B\in\bar{\mathcal K}$, this allows us to remove lower length states from the r.h.s.. The string ${\boldsymbol\sigma}^{s}_{m-3,p}$ is obtained by removing one of the bullets from ${\boldsymbol\sigma}^{s-1}_{m-3}$ and the sum over $p$ in the r.h.s. of \eqref{eq:indcondsholes2} is restricted to the positions of the bullets in ${\boldsymbol\sigma}^{s-1}_{m-3}$. Analogously we find 
\be
(\bar M^+)_{\rm red} \ket{\overline{\bullet \,{\boldsymbol\sigma}_{m-3}^{s-1}\bullet\bullet}}^{(k)} =\ket{\overline{\bullet \,{\boldsymbol\sigma}_{m-3}^{s-1}\circ \bullet}}^{(k)} +  \sum'_{p}\ket{\overline{\bullet \,{\boldsymbol\sigma}^{s}_{m-3,p}\bullet \bullet}}^{(k)}\,.
\label{eq:indcondsholes2p2}
\ee
Combining these equations we find 
\begin{align}
e^{-\frac{2\pi i k}{t}} (\bar M^+)_{\rm red} \ket{\overline{\bullet \,{\boldsymbol\sigma}_{m-3}^{s-1}\bullet\bullet}}^{(k)}+(\bar M^+)_{\rm red} \ket{\overline{\bullet \,{\boldsymbol\sigma}_{m-3}^{s-1}\bullet\bullet}}^{(k)} &= e^{-\frac{2\pi i k}{t}}\ket{\overline{\bullet \circ \,{\boldsymbol\sigma}_{m-3}^{s-1} \bullet}}^{(k)}+\ket{\overline{\bullet \,{\boldsymbol\sigma}_{m-3}^{s-1}\circ \bullet}}^{(k)}\notag\\
&+\sum'_{p}\left[e^{-\frac{2\pi i k}{t}} \ket{\overline{\bullet \,{\boldsymbol\sigma}^{s}_{m-3,p}\bullet \bullet}}^{(k)}+\ket{\overline{\bullet \bullet \,{\boldsymbol\sigma}^{s}_{m-3,p} \bullet}}^{(k)} \right]\,.
\label{eq:indcondsholes2p3}
\end{align}
The sum on the r.h.s. can be cancelled by summing \eqref{eq:indcondsholes1} for a number of appropriate ${\boldsymbol\sigma}^{s}_{m-3}$, so we have 
\begin{align}
&e^{-\frac{2\pi i k}{t}} (\bar M^+)_{\rm red} \ket{\overline{\bullet \,{\boldsymbol\sigma}_{m-3}^{s-1}\bullet\bullet}}^{(k)}+(\bar M^+)_{\rm red} \ket{\overline{\bullet \bullet \,{\boldsymbol\sigma}_{m-3}^{s-1}\bullet}}^{(k)} -\sum'_{p}(\bar P_\sigma \bar M^-)_{\rm red} \ket{\overline{\bullet \,{\boldsymbol\sigma}^{s}_{m-3,p}\bullet}}^{(k)} \notag\\
& = e^{-\frac{2\pi i k}{t}}\ket{\overline{\bullet \circ \,{\boldsymbol\sigma}_{m-3}^{s-1} \bullet}}^{(k)}+\ket{\overline{\bullet \,{\boldsymbol\sigma}_{m-3}^{s-1}\circ \bullet}}^{(k)}\,.\label{eq:indcondsholes2}
\end{align}
At this point it is useful to separate two cases: (a) $s=m-3$ and (b) $s<m-3$. In the case (b) we can find additional conditions by considering 
\be
\left[(\bar P_\sigma (\bar M^{--}))_{\rm red} -((\bar P_\sigma \bar M^-)^2)_{\rm red}\right]  \ket{\overline{\bullet \,{\boldsymbol\sigma}^{s}_{m-4}\bullet}}^{(k)}=-2e^{-\frac{2\pi i k}{t}} \ket{\overline{\bullet \bullet \,{\boldsymbol\sigma}_{m-4}^{s}\bullet\bullet}}^{(k)}\,,
\label{eq:indcondsholes3}
\ee
where we again defined $(\bar P_\sigma (\bar M^{--}))_{\rm red}$ {(\emph{cf}. Lemma~\ref{lemma:additionaloperators})} by appropriately subtracting operators in $\bar{\mathcal K}$ generating terms with smaller $m$ and $s$. Using this condition we also get 
\begin{align}
D\ket{\overline{\bullet \bullet \,{\boldsymbol\sigma}_{m-4}^{s-1}\bullet\bullet}}^{(k)}&\equiv(\bar M^+)_{\rm red} \ket{\overline{\bullet \bullet \,{\boldsymbol\sigma}_{m-4}^{s-1}\bullet\bullet}}^{(k)}+\frac{1}{2}e^{\frac{2\pi i k}{t}}\sum'_{p}\left[(\bar P_\sigma (\bar M^{--}))_{\rm red} - (\bar P_\sigma \bar M^-)^2)_{\rm red}\right]  \ket{\overline{\bullet \,{\boldsymbol\sigma}^{s}_{m-4,p}\bullet}}^{(k)} \notag\\
& =  \ket{\overline{\bullet \bullet \,{\boldsymbol\sigma}_{m-4}^{s-1}\circ\bullet}}^{(k)}+\ket{\overline{\bullet \circ \,{\boldsymbol\sigma}_{m-4}^{s-1}\bullet\bullet}}^{(k)} \,\label{eq:indcondsholes4}
\end{align}
and 
\begin{align}
&(\bar M^+)_{\rm red} \left(\ket{\overline{\bullet \bullet \,{\boldsymbol\sigma}_{m-4}^{s-2}\circ\bullet}}^{(k)}+\ket{\overline{\bullet \circ \,{\boldsymbol\sigma}_{m-4}^{s-2}\bullet \bullet}}^{(k)}\right)-\sum'_p D\ket{\overline{\bullet \bullet \,{\boldsymbol\sigma}_{m-4,p}^{s-1}\bullet\bullet}}^{(k)}\notag\\
& =  \ket{\overline{\bullet \circ \,{\boldsymbol\sigma}_{m-4}^{s-2}\circ\bullet}}^{(k)}\,,\label{eq:indcondsholes5}
\end{align}
for some $D\in \bar{\cal K}$.
To conclude the inductive step we show that in the case (b) the conditions \eqref{eq:indcondsholes1}, \eqref{eq:indcondsholes2}, \eqref{eq:indcondsholes3}, \eqref{eq:indcondsholes4} and \eqref{eq:indcondsholes5} give a complete set of vectors in the sub-space spanned by \eqref{eq:goalstates}, while in the case (a) a complete set of vectors is given by \eqref{eq:indcondsholes1}, \eqref{eq:indcondsholes2} and say \eqref{eq:indcondsholes2p1}. 
Let us start considering the case (a) and prove that the only vector in the subspace orthogonal to all vector expressions of Eqs. \eqref{eq:indcondsholes1}, \eqref{eq:indcondsholes2p1}, and \eqref{eq:indcondsholes2} is the zero vector, meaning that these equations give a complete set. We write a generic vector in the subspace as 
\be
\ket{g}=\sum_{{\boldsymbol\sigma}\in\mathcal S^{s}_{m-2}} a_{{\boldsymbol \sigma}}\ket{\overline{\bullet\, {\boldsymbol\sigma} \, \bullet}}^{(k)}\,.
\label{eq:genericvector}
\ee
where we denoted by $\mathcal S^{s}_{m}$ the set of all strings of $\bullet$ and $\circ$ with length $m$ and $s$ holes. Requiring \eqref{eq:genericvector} to be orthogonal to \eqref{eq:indcondsholes1}, \eqref{eq:indcondsholes2p1}, and \eqref{eq:indcondsholes2} we find the following conditions 
\begin{align}
&a_{{\boldsymbol\sigma}\bullet}  e^{-\frac{2\pi i k}{t}} + a_{\bullet {\boldsymbol\sigma}} =0\,,  & &\forall\, {\boldsymbol\sigma}\in\mathcal S^{m-3}_{m-3}\,, \label{eq:conda1}\\
&a_{{\boldsymbol\sigma}\circ}  e^{-\frac{2\pi i k}{t}} + a_{\circ {\boldsymbol\sigma}} =0\,,  & &\forall\, {\boldsymbol\sigma}\in\mathcal S^{m-4}_{m-3}\,,\label{eq:conda2}\\
&a_{\circ {\boldsymbol\sigma}} + a_{\bullet  \underbrace{\circ\dots\circ}_{m-3}} =0\,, & &\forall\, {\boldsymbol\sigma}\in\mathcal S^{m-4}_{m-3}\,.\label{eq:conda3}
\end{align}
Considering $\boldsymbol\sigma=\circ\bullet\overbrace{\circ\ldots\circ}^{m-4}$ in \eqref{eq:conda3} (note that since $m\geq4$ we can always consider such a string) and using \eqref{eq:conda2} we find  
\be
(1-e^{-\frac{2\pi i k}{t}})a_{\bullet \circ\dots\circ} =0
\label{eq:trueforkneq0}
\ee
which for odd $t$ gives 
\be
a_{\bullet \circ\dots\circ} =0\qquad \forall\, m\geq 4\,.
\ee
Then, using \eqref{eq:conda1} and \eqref{eq:conda2} we have that for some $n\in\mathbb N$
\be
a_{\boldsymbol \sigma} =e^{-\frac{2\pi i k n}{t}} a_{\bullet \circ\dots\circ} =0 \qquad \forall\, {\boldsymbol\sigma}\in\mathcal S^{m-3}_{m-2}\,.
\ee
Let us now consider the case (b), $s < m-3$. In this case the orthogonality conditions read as 
\begin{align}
&a_{{\boldsymbol\sigma}\bullet}  e^{-\frac{2\pi i k}{t}} + a_{\bullet {\boldsymbol\sigma}} =0\,, & &\forall\, {\boldsymbol\sigma}\in\mathcal S^{s}_{m-3}\,, \label{eq:condb1}\\
&a_{{\boldsymbol\sigma}\circ}  e^{-\frac{2\pi i k}{t}} + a_{\circ {\boldsymbol\sigma}} =0\,, & &\forall\, {\boldsymbol\sigma}\in\mathcal S^{s-1}_{m-3}\,,\label{eq:condb2}\\
& a_{\bullet \boldsymbol \sigma \bullet} =0\,, &&\forall\, {\boldsymbol\sigma}\in\mathcal S^{s}_{m-4}\,,\label{eq:condb3}\\
& a_{\circ \boldsymbol \sigma \bullet} +a_{\bullet \boldsymbol \sigma \circ} =0\,, &&\forall\, {\boldsymbol\sigma}\in\mathcal S^{s-1}_{m-4}\,,\label{eq:condb4}\\
& a_{\circ \boldsymbol \sigma \circ} =0\,, &&\forall\, {\boldsymbol\sigma}\in\mathcal S^{s-2}_{m-4}\,.\label{eq:condb5}
\end{align}
Using \eqref{eq:condb1}, \eqref{eq:condb2}, \eqref{eq:condb3}, and \eqref{eq:condb5} we see that if the string $\boldsymbol \sigma$ contains two subsequent $\bullet$ or two subsequent $\circ$ we immediately have 
\be
a_{\boldsymbol \sigma}=0\,.
\ee
So the only non zero coefficient is found for even $m$ and $s=m/2$; the associated strings have the form ${\circ\!\bullet\!\circ\dots\bullet\!\circ\bullet}$ or ${\bullet\!\circ\!\bullet\dots\circ\!\bullet\circ}$. Using \eqref{eq:condb4} we then have 
\be
a_{ \bullet \circ\bullet\circ\dots\bullet\circ\bullet\circ} =- a_{\bullet \bullet \circ\bullet\dots\circ\bullet\circ\circ} =0 \qquad\qquad a_{ \circ\bullet\circ\bullet\dots\circ\bullet\circ\bullet} =- a_{\bullet \bullet \circ\bullet\dots\circ\bullet\circ\circ} =0\,.
\ee
This proves the completeness of the conditions and concludes the inductive step.

\subsubsection{Proof of (2)}
Let us now move to consider the zero momentum sector $k=0$. First we note that the representation of the algebra $\mathcal K$ in this sector is not generically irreducible. Indeed, since $R$ commutes with all elements of ${\cal K}$ we have 
\begin{align}
&{}^{(0)}\braket{\bullet \,{\boldsymbol\sigma}'_{m'-2}\bullet|S_+ A S_- |\bullet \,{\boldsymbol\sigma}_{m-2}\bullet}^{(0)}=\notag\\
&={}^{(0)}\braket{\bullet \,{\boldsymbol\sigma}'_{m'-2}\bullet|S_+  S_- A |\bullet \,{\boldsymbol\sigma}_{m-2}\bullet}^{(0)}=0\qquad\qquad\qquad A\in\mathcal K\,,\quad m,m'\geq3\,, \forall {\boldsymbol\sigma}'_{m'-2},{\boldsymbol\sigma}_{m-2}\,.
\end{align}
In other words, elements of $\mathcal K$ cannot connect the two reflection symmetry sectors. Note that in the proof of the previous section this is reflected by the fact that for $k=0$ and $m=4$,  Eq.~ \eqref{eq:trueforkneq0} becomes trivial and does not fix anymore the coefficient $a_{\bullet \circ}$. 

This fact means that the representation of $\mathcal K$ can be irreducible only if one of the two reflection symmetry sectors becomes the zero space. The latter situation arises for small enough chain lengths as described by the following lemma
\begin{Lemma}
\label{Lemma1}
For $1<t \leq 5$ the are no vectors odd under reflection symmetry. In other words the projector $S_-$ ({cf}.~\eqref{eq:Sproj}), restricted to $V_0$, is the zero matrix. Instead, for  $t > 5$ both $S_-$ and $S_+$ are non-trivial. 
\end{Lemma}
\begin{proof}
 We note that the non-reflection symmetric vectors $\ket{\bullet \,{\boldsymbol\sigma}_{m-2}\bullet}^{(0)}$ with lowest $m$ are given by 
\begin{align}
&m=4& &\left\{\ket{\bullet \bullet \circ\, \bullet}^{(0)}\,,\qquad \ket{\bullet \circ \bullet\, \bullet}^{(0)}\right\}\,,\notag\\
&m=5& &\left\{\ket{ \bullet \circ \bullet \bullet \bullet}^{(0)}\,,\qquad \ket{\bullet \bullet \bullet \circ \bullet}^{(0)}\right\}\,.
\label{eq:nonrefsym}
\end{align}
These vectors are mapped into one another by $R$
\be
R\ket{ \bullet \circ \bullet\, \bullet}^{(0)} =  \ket{ \bullet \bullet \circ\, \bullet}^{(0)}\,,\qquad R\ket{ \bullet \circ \bullet \bullet \bullet}^{(0)} =  \ket{\bullet \bullet \bullet \circ \bullet}^{(0)}\,.
\ee
To define such states we need at least $t \geq 4$, so for $1<t \leq 3$ our claim is obvious. The fact that no reflection anti-symmetric state is present also for $t=5$ follows from the observation that for $t<6$ the states \eqref{eq:nonrefsym} with the same length are equivalent (they correspond to the same vector in $\mathcal B_0$) and are then reflection symmetric. 
\end{proof}

Our strategy in the following is to show that
\begin{itemize}
\item[(i)] Any pair of vectors even under reflection symmetry are mapped into one another by operators in $\mathcal K$.
\item[(ii)] Any pair of vectors odd under reflection symmetry are mapped into one another by operators in $\mathcal K$.
\end{itemize}
This proves that the momentum $k=0$ representation of $\mathcal K$ is irreducible for $t\leq5$, while it is split in two irreducible components for $t>5$. Note that these representations are inequivalent, moreover they are both inequivalent to all the representations in the non-zero momentum eigenspaces. This is seen by noting that the parity symmetric sector contains a representation of SU(2), generated by $M^\pm,M_z$, with spin $t/2$ while the highest weight representation of SU(2) in the antisymmetric sector has spin $t/2-3$. 

Let us start with (i) and proceed as in the previous subsection. The operator $A^+_{{\boldsymbol\sigma},{\boldsymbol\sigma}'}$, mapping the states $S_+  \ket{{\bullet \,{\boldsymbol\sigma}_{m-2}\bullet}}^{(0)}$ and $S_+  \ket{{\bullet \,{\boldsymbol\sigma}'_{m-2}\bullet}}^{(0)}$ into one another, can be written as 
\be
A^+_{{\boldsymbol\sigma},{\boldsymbol\sigma}'} = B^+_{\boldsymbol\sigma} \ket{\emptyset}\bra{\emptyset} ({B_{\boldsymbol\sigma}^+})^\dag\,.
\ee 
Here $\ket{\emptyset}$ is the ferromagnetic state with all spins down, $(\cdot)^\dag$ is the Hermitian conjugation, and $B^+_{\boldsymbol\sigma}$ is the operator mapping  $\ket{\emptyset}$ to the state $S_+  \ket{{\bullet \,{\boldsymbol\sigma}_{m-2}\bullet}}^{(0)}$. Namely
\be
B^+_{\boldsymbol\sigma} \ket{\emptyset}= S_+  \ket{\bullet{\boldsymbol\sigma}_{m-2}\bullet}^{(0)}\,.
\ee
As before we have to prove that $A^+_{{\boldsymbol\sigma},{\boldsymbol\sigma}'}\in\mathcal K$. It is easy to prove that the projector $\ket{\emptyset}\bra{\emptyset}$ is in the algebra, as it can be written as 
\be
\ket{\emptyset}\bra{\emptyset} = M^{\overbrace{+\!\ldots\!+}^{t}}M^{\overbrace{-\!\ldots\!-}^{t}}\,,
\ee
where the operator on the right hand side should be interpreted as an element of the representation of the algebra ${\cal K}$ in the 0-momentum positive-parity subspace $S_+V_0$. To prove that $A^+_{{\boldsymbol\sigma},{\boldsymbol\sigma}'} \in\mathcal K$ we then need to prove that the operator $B^+_{\boldsymbol\sigma}$ (and so also $({B^+_{\boldsymbol\sigma}})^\dag$) is in the algebra. In other words, we have to prove that we can map $\ket{\overline{\emptyset}}$ to any state in the symmetric sector by using elements of the algebra. Again, to prove this statement we consider the basis (\ref{eq:statesb}) of the double sized ($2t$) chain, including also the ferromagnetic states $\ket{\overline{\emptyset}}$ and $\ket{\overline{{\bullet\ldots\bullet}}}$ in the basis. 

First we note that Eqs.~\eqref{eq:mapamongblocks} and \eqref{eq:m3connection} imply that we can map $\ket{\overline \emptyset}$ into any state of the form $\bar S_+  \ket{\overline{\bullet \,{\boldsymbol\sigma}_{m-3}\bullet}}^{(0)}$ with length $m\leq3$ and into any block state \eqref{eq:blockstates} (which are even under reflection).

Let us now prove that acting with operators in $\bar{\mathcal K}$ on the state $\ket{\overline \emptyset}$ we can access every state of the form  
\be
\bar S_+  \ket{\overline{\bullet \,{\boldsymbol\sigma}^s_{m-2}\bullet}}^{(0)}\,,\qquad s=0,\ldots,m-2\,,\qquad m=3,\ldots, t\,,
\label{eq:goalstatesS+}
\ee
assuming that we can access every state $\bar S_+\ket{\overline{\bullet \,{\boldsymbol\sigma}^{s}_{n-3}\bullet}}^{(0)}$ of length $n<m$ and every state $\bar S_+\ket{\overline{\bullet \,{\boldsymbol\sigma}_{m-2}^{r}\bullet}}^{(0)}$ with $r<s$. Using the relations \eqref{eq:indcondsholes1}, \eqref{eq:indcondsholes2p1}, and \eqref{eq:indcondsholes2} we have 
\begin{align}
& A \ket{\overline\emptyset} = \bar S_+\ket{\overline{\bullet \bullet \,{\boldsymbol\sigma}_{m-3}^{s} \bullet}}^{(0)}+\bar S_+\ket{\overline{\bullet \,{\boldsymbol\sigma}_{m-3}^{s}\bullet \bullet}}^{(0)}\,,  \label{eq:S+indcondsholes1}\\
& B \ket{\overline\emptyset} =\bar S_+\ket{\overline{\bullet \circ \,{\boldsymbol\sigma}_{m-3}^{s-1} \bullet}}^{(0)}+\bar S_+\ket{\overline{\bullet \,{\boldsymbol\sigma}_{m-3}^{s-1}\circ \bullet}}^{(0)}\,, \label{eq:S+indcondsholes2}\\
&C \ket{\overline\emptyset} =  \bar S_+\ket{\overline{\bullet \circ \,{\boldsymbol\sigma}^{s-1}_{m-3}\bullet}}^{(0)}+ \sum'_p   \bar S_+\ket{\overline{\bullet\bullet \,{\boldsymbol\sigma}^s_{m-3,p}\bullet}}^{(0)}  \,,\qquad\qquad A,B,C\in \bar{\mathcal K}\,.\label{eq:S+indcondsholes3}
\end{align} 
As we did before we distinguish two cases: (a) $s=m-3$ and (b) $s<m-3$. Let us prove that in the case (a) Eqs.~\eqref{eq:S+indcondsholes1},  \eqref{eq:S+indcondsholes2} and \eqref{eq:S+indcondsholes3} form a complete set of vectors in the sub-space spanned by \eqref{eq:goalstatesS+}. To do that we show that a generic linear combination of the vectors \eqref{eq:goalstatesS+} is orthogonal to all \eqref{eq:S+indcondsholes1},  \eqref{eq:S+indcondsholes2}, and \eqref{eq:S+indcondsholes3} only if is the zero vector. The most general linear combination of states \eqref{eq:goalstatesS+} can be written as 
\begin{align}
\ket{g}_{\rm sym}=  &\sum'_{{\boldsymbol\sigma}\in\mathcal S^{s}_{m-2}} a_{{\boldsymbol\sigma}} \ket{\overline{ \bullet \, {\boldsymbol\sigma} \bullet}}^{(0)}\,,
\label{eq:genericvectorS+}
\end{align}
where $a_{\boldsymbol \sigma}$ is subject to the constraint 
\be
a_{{\boldsymbol\sigma}}=a_{{\boldsymbol\sigma}^R}\,,
\label{eq:constraint}
\ee
and $\cdot^R$ denotes the reflection operation (reversal) of the string of bullets and holes. The orthogonality conditions read
\begin{align}
&a_{{\boldsymbol\sigma}\bullet}  + a_{\bullet {\boldsymbol\sigma}} =0\,, \qquad \forall\, {\boldsymbol\sigma}\in\mathcal S^{m-3}_{m-3}\,, \label{eq:conda1S+}\\
&a_{{\boldsymbol\sigma}\circ} + a_{\circ {\boldsymbol\sigma}} =0\,, \qquad \forall\, {\boldsymbol\sigma}\in\mathcal S^{m-4}_{m-3}\,.\label{eq:conda2S+}\\
&a_{\circ {\boldsymbol\sigma}} + a_{\bullet \underbrace{\circ\dots\circ}_{m-3}}  =0\,, \qquad \forall\, {\boldsymbol\sigma}\in\mathcal S^{m-4}_{m-3}\,.\label{eq:conda3S+}
\end{align}
We see that now the case $m=4$ is fixed by either \eqref{eq:conda1S+} or \eqref{eq:conda3S+}, indeed using \eqref{eq:constraint} we have 
\be
a_{\circ\bullet}+a_{\bullet\circ}=2a_{\circ\bullet}=2a_{\bullet\circ}=0\,.
\ee
For $m>4$ we choose $\boldsymbol\sigma=\circ\circ\bullet\overbrace{\circ\ldots\circ}^{m-5}$ in \eqref{eq:conda3S+}, and using \eqref{eq:conda2S+} two times we find  
\be
a_{\bullet \circ \ldots\circ} = 0 \,.
\ee
Then, using \eqref{eq:conda1S+} and \eqref{eq:conda2S+} we have 
\be
a_{\boldsymbol \sigma} = a_{\bullet \circ\dots\circ} =0 \qquad \forall\, {\boldsymbol\sigma}\in\mathcal S^{m-3}_{m-2}\,.
\ee
The case (b) is treated by considering the additional conditions originating from \eqref{eq:indcondsholes3}, \eqref{eq:indcondsholes4} and \eqref{eq:indcondsholes5}, namely
\begin{align}
&D \ket{\bar{\emptyset}} = \bar S_+ \ket{\overline{\bullet \bullet \,{\boldsymbol\sigma}_{m-4}^{s}\bullet\bullet}}^{(0)}\,,
\label{eq:S+indcondsholes3}\\
& E \ket{\bar{\emptyset}} = \bar  S_+\ket{\overline{\bullet \bullet \,{\boldsymbol\sigma}_{m-4}^{s-1}\circ\bullet}}^{(0)}+\bar S_+ \ket{\overline{\bullet \circ \,{\boldsymbol\sigma}_{m-4}^{s-1}\bullet\bullet}}^{(0)} \,\label{eq:S+indcondsholes4}\\
& F \ket{\bar{\emptyset}} = \bar S_+ \ket{\overline{\bullet \circ \,{\boldsymbol\sigma}_{m-4}^{s-2}\circ\bullet}}^{(0)}\,,\qquad\qquad\qquad\qquad\qquad\qquad D,E,F\in\bar{\mathcal K}\,.\label{eq:S+indcondsholes5}
\end{align}
The orthogonality conditions read as 
\begin{align}
&a_{{\boldsymbol\sigma}\bullet}+ a_{\bullet {\boldsymbol\sigma}} =0\,, & &\forall\, {\boldsymbol\sigma}\in\mathcal S^{s}_{m-3}\,, \label{eq:condb1S+}\\
&a_{{\boldsymbol\sigma}\circ}   + a_{\circ {\boldsymbol\sigma}} =0\,, & &\forall\, {\boldsymbol\sigma}\in\mathcal S^{s-1}_{m-3}\,,\label{eq:condb2S+}\\
& a_{\bullet \boldsymbol \sigma \bullet} =0\,, &&\forall\, {\boldsymbol\sigma}\in\mathcal S^{s}_{m-4}\,,\label{eq:condb3S+}\\
& a_{\circ \boldsymbol \sigma \bullet} +a_{\bullet \boldsymbol \sigma \circ} =0\,, &&\forall\, {\boldsymbol\sigma}\in\mathcal S^{s-1}_{m-4}\,,\label{eq:condb4S+}\\
& a_{\circ \boldsymbol \sigma \circ} =0\,, &&\forall\, {\boldsymbol\sigma}\in\mathcal S^{s-2}_{m-4}\,.\label{eq:condb5S+}
\end{align}
and are solved as in the above section giving 
\be
a_{\boldsymbol \sigma} =0 \qquad \forall\, {\boldsymbol\sigma}\in\mathcal S^{s}_{m-2}\,, \qquad \forall\,s<m-2.
\ee
Let us now consider the point (ii).  In this case, the operator $A^-_{{\boldsymbol\sigma},{\boldsymbol\sigma}'}$, mapping the states $S_-  \ket{{\bullet \,{\boldsymbol\sigma}_{m-2}\bullet}}^{(0)}$ and $S_-  \ket{{\bullet \,{\boldsymbol\sigma}'_{m-2}\bullet}}^{(0)}$ into one another, can be written as 
\be
A^-_{{\boldsymbol\sigma},{\boldsymbol\sigma}'} = B^-_{\boldsymbol\sigma} S_- \ket{{\bullet \circ \bullet\, \bullet}}^{(0)(0)}\!\bra{{\bullet \circ \bullet\, \bullet}}S_- ({B_{\boldsymbol\sigma}^-})^\dag\,.
\ee 
Here $S_-\ket{{\bullet \circ \bullet\, \bullet}}^{(0)}$ is the parity-odd state with minimal length and $B^-_{\boldsymbol\sigma}$ is the operator mapping $S_-\ket{{\bullet \circ \bullet\, \bullet}}^{(0)}$ to the state $S_-  \ket{{\bullet \,{\boldsymbol\sigma}_{m-2}\bullet}}^{(0)}$. Namely
\be
B^-_{\boldsymbol\sigma} S_-\ket{{\bullet \circ \bullet\, \bullet}}^{(0)} = S_-  \ket{\bullet{\boldsymbol\sigma}_{m-2}\bullet}^{(0)}\,.
\ee
Once again, we need to prove that $A^-_{{\boldsymbol\sigma},{\boldsymbol\sigma}'}\in\mathcal K$. The projector $S_- \ket{{\bullet \circ \bullet\, \bullet}}^{(0)(0)}\bra{{\bullet \circ \bullet\, \bullet}}S_-$ can be written as 
\be
S_- \ket{{\bullet \circ \bullet\, \bullet}}^{(0)(0)}\bra{{\bullet \circ \bullet\, \bullet}}S_-=  e^{i \frac{\pi}{2} (M^x-t)} M^{\overbrace{z\!-\!\ldots\!-\!z}^{t-4}} e^{i \frac{\pi}{2} (M^x-t)} P_- M^{\overbrace{z\!-\!\ldots\!-\!z}^{t-4}}\,,
\ee
and it is then an element of the algebra as a consequence of Lemma~\ref{lemma:additionaloperators}. The operator on the right hand side should be interpreted as an element of the representation of the algebra ${\cal K}$ in the 0-momentum negative-parity subspace $S_-V_0$. So, as before, to show that $A^-_{{\boldsymbol\sigma},{\boldsymbol\sigma}'} \in  \mathcal K$ we need to show that $B^-_{\boldsymbol\sigma}$ (and so also $(B^-_{\boldsymbol\sigma})^\dag$) is in the algebra, \emph{i.e.}, that acting with operators of $\mathcal K$ on $S_- \ket{{\bullet \circ \bullet\, \bullet}}^{(0)}$ we can access every state in the odd parity sector. Such states have the form  
\be
S_-  \ket{\bullet \,{\boldsymbol\sigma}^s_{m-2}\bullet}^{(0)}\,,\qquad s=1,\ldots,m-3\,,\qquad m=3,\ldots, t\, ,
\label{eq:goalstatesS-}
\ee
where we considered values of $s$ giving non-vanishing states.

We proceed once again by induction and consider the corresponding basis states (\ref{eq:statesb}) of the double sized ($2t$) chain and finally use the projector $Q$ (\ref{eq:Q}) to obtain all the required mappings. Let us start by constructing the basis for our inductive construction. Applying $ \bar P_+ \bar M^-$ on $\ket{ \overline{\bullet \circ \bullet\, \bullet}}^{(0)}$ we have  
\be
\bar P_+ \bar M^- \bar S_- \ket{\overline{\bullet \circ \bullet\, \bullet}}^{(0)} = \bar S_- \ket{\overline{\bullet \circ \bullet \bullet \bullet}}^{(0)}\,,
\ee
which is the only antisymmetric state of length $5$ with one hole. We also have 
\be
\left((\bar P_+ \bar M^- \bar P_+ \bar M^+) -\1 \right)\bar S_- \ket{\overline{\bullet \circ \bullet \,\bullet}}^{(0)} = \bar S_- \ket{\overline{\bullet \circ \circ \bullet \bullet}}^{(0)} \,,
\ee
which is the only antisymmetric state of length $5$ and two holes. So we explicitly constructed any antisymmetric state  \eqref{eq:goalstatesS-} of length $m\leq 5$. 

Let us now show that if we can access every state $\bar S_-\ket{\overline{\bullet \,{\boldsymbol\sigma}^{s}_{n-3}\bullet}}^{(0)}$ of length $n<m$ and every state $\bar S_-\ket{\overline{\bullet \,{\boldsymbol\sigma}_{m-2}^{r}\bullet}}^{(0)}$ with $r<s$ then we can access all the states \eqref{eq:goalstatesS-} for $m\geq6$. 
To do this we need to distinguish three cases: (a) $s=1$; (b) $1<s<m-3$; (c) $s=m-3$. Let us start form the case (c): if we can get all the states $\bar S_-\ket{\overline{\bullet \,{\boldsymbol\sigma}_{m-2}^{1}\bullet}}^{(0)}$ then we consider 
\be
(\bar P_+ \bar M^+)^{m-4} \bar S_-\ket{\overline{\bullet \circ \underbrace{\bullet \ldots \bullet}_{m-2}}}^{(0)} - \bar S_-\ket{\overline{\bullet \circ \bullet\, \bullet}}^{(0)} = \bar S_-\ket{\overline{\bullet \underbrace{\circ \ldots \circ}_{m-3} \bullet\, \bullet}}^{(0)}\,,
\ee
giving a vector with length $m$ and $m-3$ holes. The other conditions are found by projecting \eqref{eq:indcondsholes2} on the parity odd space 
\be
A \bar S_- \ket{\overline{ \bullet \circ \bullet\, \bullet}}^{(0)} = \bar S_-\ket{\overline{\bullet \circ \,{\boldsymbol\sigma}_{m-3}^{m-4} \bullet}}^{(0)}+\bar S_-\ket{\overline{\bullet \,{\boldsymbol\sigma}_{m-3}^{m-4}\circ \bullet}}^{(0)}\,,\qquad A\in \bar{\mathcal K}\,.\\
\ee
Let us show that these conditions give a complete set of vectors in the parity odd subspace. Once again we will prove that by showing that only the zero vector is orthogonal to all of the vectors expressing the conditions. Considering a generic vector of length $m$ and $s$ holes in the parity odd subspace with we have 
\begin{align}
\ket{g}_{\rm antisym}=  &\sum''_{{\boldsymbol\sigma}\in\mathcal S^{s}_{m-2}} a_{{\boldsymbol\sigma}} \ket{ \overline{\bullet \, {\boldsymbol\sigma} \bullet}}^{(0)}\,,
\label{eq:genericvectorS-}
\end{align} 
where the coefficients are subject to the constraint 
\be
a_{{\boldsymbol\sigma}}=-a_{{\boldsymbol\sigma}^R}\,.
\label{eq:constraintS-}
\ee
The orthogonality conditions read as 
\begin{align}
&a_{{\boldsymbol\sigma}\circ}+ a_{\circ {\boldsymbol\sigma}} =0\,, & &\forall\, {\boldsymbol\sigma}\in\mathcal S^{m-4}_{m-3}\,, \label{eq:condc1S-}\\
&a_{\underbrace{\circ \ldots \circ}_{m-3} \bullet}=0\,. \label{eq:condc2S-}
\end{align}
Since using \eqref{eq:condc1S-} we can always bring $a_{{\boldsymbol\sigma}}$ in the form \eqref{eq:condc2S-} we have 
\be
a_{{\boldsymbol\sigma}}=0\,,\qquad\qquad\qquad\forall\, {\boldsymbol\sigma}\in\mathcal S^{m-3}_{m-2}\,.
\ee
Let us now consider the case (a). Here we cannot use conditions from the states with $s-1$ holes since there are none. We then use 
\begin{align}
& A \bar S_- \ket{\overline{\bullet \circ \bullet\, \bullet}}^{(0)} = \bar S_-\ket{\overline{\bullet \bullet \,{\boldsymbol\sigma}_{m-3}^{1} \bullet}}^{(0)}+\bar S_-\ket{\overline{\bullet \,{\boldsymbol\sigma}_{m-3}^{1}\bullet \bullet}}^{(0)}\,, \label{eq:S-indcondsholes1}\\
&B \bar S_- \ket{\overline{\bullet \circ \bullet\, \bullet}}^{(0)} = \bar S_- \ket{\overline{\bullet \bullet \,{\boldsymbol\sigma}_{m-4}^{s}\bullet\bullet}}^{(0)}\,,\qquad\qquad\qquad\qquad\qquad A,B\in\bar{\mathcal K}\,.
\label{eq:S-indcondsholes3}
\end{align} 
The orthogonality condition then reads as 
\begin{align}
&a_{{\boldsymbol\sigma}\bullet}+ a_{\bullet {\boldsymbol\sigma}} =0\,, & &\forall\, {\boldsymbol\sigma}\in\mathcal S^{1}_{m-3}\,, \label{eq:condb1S-}\\
&a_{\bullet {\boldsymbol\sigma}\bullet}=0\,, & &\forall\, {\boldsymbol\sigma}\in\mathcal S^{1}_{m-4}\,. \label{eq:condb1S-}
\end{align}
with the constraint \eqref{eq:constraintS-}. These conditions immediately fix
\be
a_{{\boldsymbol\sigma}} =0\,, \qquad\qquad \forall\, {\boldsymbol\sigma}\in\mathcal S^{1}_{m-2}\,.
\ee
Finally, let us consider the case (b). It is convenient to subdivide it in two additional cases: (b1) $1<s<m-4$; (b2) $s=m-4$. The case (b1) is treated exactly as in the previous sections. Let us consider the case (b2). In this case we have the following orthogonality conditions 
\begin{align}
&a_{{\boldsymbol\sigma}\bullet}  + a_{\bullet {\boldsymbol\sigma}} =0\,,  &&\forall\, {\boldsymbol\sigma}\in\mathcal S^{m-4}_{m-3}\,, \label{eq:condb21S-}\\
&a_{{\boldsymbol\sigma}\circ} + a_{\circ {\boldsymbol\sigma}} =0\,,  &&\forall\, {\boldsymbol\sigma}\in\mathcal S^{m-5}_{m-3}\,.\label{eq:condb22S-}\\
&a_{\circ {\boldsymbol\sigma}} + \sum_p a_{\bullet {\boldsymbol\sigma}_p}  =0\,, &&\forall\, {\boldsymbol\sigma}\in\mathcal S^{m-5}_{m-3}\,.\label{eq:condb23S-}
\end{align}
together with the constraint \eqref{eq:constraintS-}. We remind the reader that the sum in \eqref{eq:condb23S-} is over the positions of the two bullets in $\boldsymbol \sigma$ and $\boldsymbol \sigma_p$ is the string obtained from $\boldsymbol \sigma$ by replacing the bullet at position $p$ with a hole. 

A generic coefficient $a_{\boldsymbol\sigma}$ with $\boldsymbol\sigma$ in ${\cal S}^{m-4}_{m-2}$ can be written as 
\be
a_{\boldsymbol\sigma} = a_{\underbrace{\circ\ldots\circ}_{\ell_1}\bullet\underbrace{\circ\ldots\circ}_{\ell_2}\bullet\underbrace{\circ\dots\circ}_{\ell_3}}
\ee 
Using the relations \eqref{eq:condb21S-}, \eqref{eq:condb22S-} and the constraint \eqref{eq:constraintS-} we see that it can be non zero only if $\ell_1+\ell_3$ and $\ell_2$ are both odd. In this case, however, we consider \eqref{eq:condb23S-} with $\boldsymbol\sigma=\underbrace{\circ\ldots\circ}_{\ell_1-1}\bullet\underbrace{\circ\ldots\circ}_{\ell_2}\bullet\underbrace{\circ\dots\circ}_{\ell_3}$. This gives 
\be
0=a_{\underbrace{\circ\ldots\circ}_{\ell_1}\bullet\underbrace{\circ\ldots\circ}_{\ell_2}\bullet\underbrace{\circ\dots\circ}_{\ell_3}} +a_{\bullet \underbrace{\circ\ldots\circ}_{\ell_1+\ell_2-1}\bullet\underbrace{\circ\dots\circ}_{\ell_3}}+a_{\bullet\underbrace{\circ\ldots\circ}_{\ell_1-1}\bullet\underbrace{\circ\ldots\circ}_{\ell_2+\ell_3}} =a_{\underbrace{\circ\ldots\circ}_{\ell_1}\bullet\underbrace{\circ\ldots\circ}_{\ell_2}\bullet\underbrace{\circ\dots\circ}_{\ell_3}}\,.
\ee
where we used that if $\ell_1+\ell_3$ and $\ell_2$ are odd, both the sets $\{\ell_1-1,\ell_2+\ell_3\}$ and $\{\ell_1+\ell_2-1,\ell_3\}$ contain an even number. We then have 
\be
a_{{\boldsymbol\sigma}} =0\,, \qquad\qquad \forall\, {\boldsymbol\sigma}\in\mathcal S^{m-4}_{m-2}\,,
\ee
concluding the proof of (2). 

\subsubsection{Proof of (3)}
To prove the point (3) we show that  
\be
{}^{(k)}\braket{\bullet|M^+|\bullet\!\bullet}^{(k)}
\label{eq:relevantmatrixelement}
\ee
has different absolute value for different values of momentum $k\in\{1,2,\ldots (t-1)/2\}$. This is enough because the states $\ket{\bullet}^{(k)}$ and $\ket{\bullet \bullet}^{(k)}$ are unique up to a phase. Namely they are the only states in the representation such that 
\be
M_z\ket{\bullet}^{(k)}= (t-2)\ket{\bullet}^{(k)}\,, \qquad\qquad\qquad U\ket{\bullet}^{(k)} =-\ket{\bullet}^{(k)} \,.
\ee
\be
M_z\ket{\bullet \bullet}^{(k)}= (t-4)\ket{\bullet \bullet}^{(k)}\,, \qquad\qquad\qquad U\ket{\bullet \bullet}^{(k)} =-\ket{\bullet\bullet}^{(k)} \,.
\ee
This means that if \eqref{eq:relevantmatrixelement} has different absolute value for two representations with momenta $k$ and $p$, there cannot be any transformation $C$ fulfilling \eqref{eq:point3}. Bringing $(M^+)_k$ and $(M^+)_p$ to the same form either makes different the form of  $(M_z)_k$ and $(M_z)_p$ or that of $(U)_k$ and $(U)_p$. A direct calculation gives 
\be
{}^{(k)}\braket{\bullet|M^+|\bullet\!\bullet}^{(k)}=1+ e^{\frac{2\pi i k}{t}}\,.
\ee
So we have that ${}^{(k)}\braket{\bullet|M^+|\bullet\!\bullet}^{(k)}$ and ${}^{(p)}\braket{\bullet|M^+|\bullet\!\bullet}^{(p)}$ have the same absolute value if 
\be
|\cos\frac{\pi k}{t}|=|\cos\frac{\pi p}{t}|\,,\qquad\qquad k,p=1,...,t-1\,,
\ee
which is fulfilled only for $p=k$ and $p=t-k$. This concludes the proof. }

\section{Results in the even $t$ case}
\label{app:evencase}

Here we state a few properties and observations relevant for the case of {\em even} times $t$. We start by showing the following simple Lemma (supplementing Lemma~\ref{Lemma1} of the previous section)

\begin{Lemma}
\label{Lemma2}

For even $t$ and $t \ge 4$ both projectors $S_\pm$ ({cf}.~\eqref{eq:Sproj}), restricted to momentum $k=t/2$ subspace $V_{t/2}$,  
are non-trivial (non-zero matrices). For $t=2$, $V_1$ is one-dimensional and the projector $S_+$ is non-trivial (with rank 1), while the
projector $S_-$ is a zero matrix.
\end{Lemma}
\begin{proof}
For $t\ge 4$, examples of non-vanishing elements of $S_\pm V_{t/2}$ are $S_\pm \ket{\bullet \circ \bullet\, \bullet}^{(t/2)}$. For $t=2$, we have a single basis element of $V_1$, namely $v=\ket{\bullet\,\circ}^{(1)}$, for which $S_+ v = v$, $S_- v=0$.
\end{proof}
As a consequence, $Y_{t/2}$ and $Y'_{t/2}$ are linearly independent for all even $t\ge 4$, and $Y_1+Y'_1=0$ for $t=2$.
\\\\
Let us now state a few observations concerning eigenvectors of $\mathbb T$ of eigenvalues $\pm1$ for even $t$.

\subsubsection{Additional eigenvectors with eigenvalue 1}
We start by enumerating eigenvectors of eigenvalue $1$, or equivalently, linearly independent operators $A$ simultaneously commuting with $\{M_x,M_y,M_z,U\}$.
Besides to the set of $2t$ operators $\{ Y_k,Y_k';k=0,1,\ldots,t-1\}$, all having ranks which grow exponentially with $t$ (typically as 
$\sim \frac{1}{2t} 2^t$) we generally have an additional rank one operator
\be
Z=\ket{\psi}\bra{\psi}
\label{eq:Z}
\ee
such that $Z^2 = Z$. Here the normalized state $\ket{\psi}$, is a particular scrambled fermi sea
\be
\ket{\psi}=\frac{1}{2^t} \prod_{\tau =1}^{t/2} \left(1-P_{\tau,\tau+t/2}\right)\ket{\underbrace{\da,\ldots,\da}_{t/2},\underbrace{\ua,\ldots,\ua}_{t/2}}\,.
\label{eq:psi}
\ee
It is straightforward to chek that this state is a spin singlet $M_\alpha \ket{\psi}=0$. Moreover, it is an eigenvector of $U,\Pi,R$
satisfying ${U\ket{\psi}=-\ket{\psi}}$, $\Pi \ket{\psi}=-\ket{\psi}$, $R\ket{\psi}=(-1)^{t/2} \ket{\psi}$. Thus, $Z$ satisfies Property 2 with $\phi=0$, i.e. it commutes with $\{U,M_x,M_y,M_z\}$, and is orthogonal to (hence linearly independent from) all other general solutions 
\be
Z Y_k = Y_k Z = Z Y'_k = Y'_k Z = 0,\quad \forall k.
\ee
The states $Z$, for all even $t$, have momentum $k=t/2$ and are odd under reflection in the sense, $Z \Pi = \Pi Z = -Z$, $Z R = R Z = - Z$. This operator is generically the only additional operator that we identify. For $t=8$ and $t=10$, however, we find a further additional solution in the zero momentum sector.
\\\\
For $t=8$, the additional solution is a rank-3 projector
\be
Z' = \sum_{\beta=-1}^1 \ket{\psi_\beta}\bra{\psi_\beta},
\ee
where $\{\psi_\beta\}$ is a triplet of states
\begin{eqnarray}
\ket{\psi_1} &=& \sqrt{2} Y_0 \left( \ket{\ua\ua\da\ua\da\da\da\da}-\ket{\ua\ua\da\da\ua\da\da\da}+\ket{\ua\ua\da\da\da\ua\da\da}-\ket{\ua\ua\da\da\da\da\ua\da}\right), \nonumber \\
\ket{\psi_0} &=& 2 Y_0 \left( \ket{\ua\ua\da\ua\da\ua\da\da} - \ket{\ua\ua\da\da\ua\da\ua\da} \right), \\
\ket{\psi_{-1}} &=& \sqrt{2} Y_0 \left( \ket{\da\da\ua\da\ua\ua\ua\ua}-\ket{\da\da\ua\ua\da\ua\ua\ua}+\ket{\da\da\ua\ua\ua\da\ua\ua}-\ket{\da\da\ua\ua\ua\ua\da\ua}\right). \nonumber
\end{eqnarray}
One can check straightforwardly that $M_z \ket{\psi_\beta} = -2 \beta \ket{\psi_\beta}$,
$(M_x^2+M_y^2+M_z^2)\ket{\psi_\beta} = 8 \ket{\psi_\beta}$, meaning that $\ket{\psi_\beta}$ indeed form a spin-1 triplet, and $U$ acts as a parity
\be
U\ket{\psi_\beta} = -(-1)^\beta \ket{\psi_\beta}.
\ee
This solution, again normalized such that $Z'^2 = Z'$, is independent from all the other solutions 
\be
ZZ'=Z'Z=Z' Y_k = Y_k Z' = Z' Y'_k = Y'_k Z' = 0,\quad \forall k,
\ee
and satisfies $\Pi Z' = Z' \Pi  = Z'$, $R Z' = Z' R = -Z'$.
\\\\
For $t=10$, the additional solution is again a rank-1 projector,
\be
Z'' = \ket{\chi}\bra{\chi}
\ee
normalized as $Z''^2 = Z''$, with
\be
\ket{\chi}=\frac{1}{\sqrt{30}}S_- Y_0\bigl(\ket{\ua\ua\ua\ua\da\ua\da\da\da\da}-\ket{\ua\ua\ua\ua\da\ua\da\da\da\da}-\ket{\ua\ua\ua\da\ua\ua\da\da\da\da}  + 
\ket{\ua\ua\ua\da\da\ua\ua\da\da\da}-\ket{\ua\ua\da\ua\da\ua\da\ua\da\da}-\ket{\ua\ua\da\ua\da\ua\da\da\ua\da}\bigr),
\label{eq:chi}
\ee
again being a spin singlet $M_\alpha\ket{\chi}=0$, and an eigenvector of $U,\Pi,R$, satisfying $U\ket{\chi}=\ket{\chi}$, $\Pi\ket{\chi} = \ket{\chi}$, $R\ket{\chi} = -\ket{\chi}$.
As in the above cases, it is easy to check the linear independence (orthogonality) with respect to the other solutions
\be
ZZ''=Z''Z=Z'' Y_k = Y_k Z'' = Z'' Y'_k = Y'_k Z'' = 0,\quad \forall k.
\ee
For higher even $t$ we have found no additional solutions besides (\ref{eq:Z}), hence we conjecture that for $t\ge 11$ we have exactly $2t+1$ independent eigenvectors of $\mathbb T$ with eigenvalue $1$.

\subsubsection{Eigenvectors with eigenvalue -1}

As we have shown, eigenvectors of $\mathbb T$ of eigenvalue -1 are only possible for even $t$. Here, however, we argue that also for even $t$ eigenvalues -1 are exceptional and can appear only in finitely many cases. For example, a linearly independent pair of operators $A_\pm$ satisfying Property 2 with $\phi=\pi$ can be obtained with the ansatz
\be
A_+ = \ket{\psi_+}\bra{\psi_-},\qquad
A_- = \ket{\psi_-}\bra{\psi_+},
\ee
where vectors $\ket{\psi_\pm}$ are both spin singlets $M_\alpha \ket{\psi_\pm}=0$, and have opposite eigenvalues of $U$, $U\ket{\psi_\pm} = \pm \ket{\psi_\pm}$.

This happens in two cases
\begin{itemize}
\item[(i)] $t=6$; where the projector $S_- Y_0$ has rank one so it can be written as $S_- Y_0 = \ket{\psi_+}\bra{\psi_+}$, while the second state in the pair, $\ket{\psi_-}$, is given by (\ref{eq:psi}).
\item[(ii)] $t=10$; where the first state $\ket{\psi_+}=\ket{\chi}$ is given by (\ref{eq:chi}), while the second state in the pair, $\ket{\psi_-}$, is again given by (\ref{eq:psi}).
\end{itemize}
For higher even $t$ up to $t=16$, we have found no other rank-1 operators commuting with $\{M_x,M_y,M_z,U\}$, besides (\ref{eq:Z}), so no other eigenvalue $-1$ eigenoperators can be constructed. Thus we conjecture that there are no other eigenvectors of eigenvalue $-1$.

Note that the explicit cases discussed above, together with Property 4, completely explain the empirical Table~I reported in the main text. 

\section{Diagonalization of the Integrable (transverse field) Kicked Ising Model}
\label{app:integrablecase}

In this section we consider a non-trivial integrable limit of the Floquet operator (3), namely the case $\boldsymbol h=0$. In this case this operator can be written in terms of a free fermionic Hamiltonian. One defines the fermionic operators through a Jordan-Wigner transformation as follows 
\be
c_l^\dag = \prod_{j=1}^{l-1} \sigma_j^x \sigma_l^-\,,\qquad\qquad\qquad c_l = \prod_{j=1}^{l-1} \sigma_j^x \sigma_l^+\,.
\ee
Here we introduced
\be
\sigma_j^-=\frac{1}{2}(\sigma_j^z-i\sigma_j^y)\,,\qquad\qquad\qquad\sigma_j^+=\frac{1}{2}(\sigma_j^z+i\sigma_j^y)\,.
\ee
In terms of these fermions the Floquet operator can be written as  
\be
U_{\rm KI}[0]= e^{i H_{\rm eff}(J,b)}\,,
\ee
with 
\be
H_{\rm eff}(J,b)=\frac{e^{i \pi N}+1}{2}H^{\rm e}(J,b) \frac{e^{i \pi N}+1}{2} + \frac{e^{i \pi N}-1}{2}H^{\rm o}(J,b) \frac{e^{i \pi N}-1}{2}\,,\qquad\qquad N=\sum_{i=1}^L c^\dag_ic_i\,.
\label{eq:Heff}
\ee
Here $({e^{i \pi N}\pm1})/{2}$ respectively projects on the sectors with even and odd number of particles. Note that in terms of the spin operators $\{\sigma_i^\alpha\}$ we have 
\be
e^{i \pi N}=\prod_{j=1}^L  (1-2 c^\dag_jc_j)=\prod_{j=1}^L  \sigma_j^x\,,
\ee
so $e^{i \pi N}$ corresponds to a spin-flip transformation. The Hamiltonians $H^{\rm e}(J,b), H^{\rm o}(J,b) $ can be diagonalised by a Bogoliubov transformation: the final result reads as
\be
H^{\rm e(o)}(J,b) =\sum_{k\in\rm NS(R)} \epsilon(k) \left(b^\dag_kb^{\phantom{\dag}}_k-\frac{1}{2}\right)\,,
\ee
with $\rm NS(R)$ representing Neveu-Schwartz (Ramond) sectors
\begin{align}
k\in{\rm R}\qquad &\Rightarrow\qquad k=\frac{2\pi}{L}n\,,& &n \in\mathbb{Z}\cap[0,L[\,,\\
k\in{\rm NS}\qquad &\Rightarrow\qquad k=\frac{2\pi}{L}\left(n+\frac{1}{2}\right)\,,& &n \in\mathbb{Z}\cap[0,L[\,,
\end{align}
 and 
\be
\epsilon(k)\equiv -\cos^{-1}(\cos(2b)\cos(2J)+\cos(k)\sin(2J)\sin(2b))\!\!\!\!\mod 2\pi \,.
\ee
A complete basis of the Hilbert space is constructed by defining two states 
\be
\ket{0}_{\rm NS},\quad\ket{0}_{\rm R}\,,
\ee 
such that 
\be
b_k \ket{0}_{\rm R}=0\quad\textrm{if}\quad k\in{\rm R},\qquad\qquad\qquad\qquad b_k \ket{0}_{\rm NS}=0\quad\textrm{if}\quad k\in{\rm NS}\,.
\ee
For $L$ odd we have 
\begin{align}
e^{i \pi N} \ket{0}_{\rm NS}&=\text{sign}(b+J)\ket{0}_{\rm NS}\,, & e^{i \pi N} \ket{0}_{\rm R}&=\text{sign}(b-J)\ket{0}_{\rm R}\,,\\
\Pi_L \ket{0}_{\rm NS}&=(-1)^{\theta(-b-J)}\ket{0}_{\rm NS}\,, & \Pi_L \ket{0}_{\rm R}&= \ket{0}_{\rm R}\,,
\end{align}
where $\Pi_L$ is the shift operator (\emph{cf}. Eq. (21) of the main text) on a chain of $L$ sites and $\theta(x)$ is the step function. The basis is constructed acting with the operators $b^\dag_k$ on $\ket{0}_{\rm NS},\ket{0}_{\rm R}$ keeping only the vectors in the R sector with negative eigenvalue of $e^{i \pi N}$ and those in the NS sector with positive eigenvalue of $e^{i \pi N}$. In the case $b <-|J|$ we find 
\be
\ket{k_1, \ldots, k_{2n+1}}_{\rm NS} = \prod_{\substack{i=1\\ \\  k_i \in \rm NS}}^{2n+1} b^{\dagger}_{k_i}\ket{0}_{\rm NS}\,,\quad \ket{k_1, \ldots, k_{2m}}_{\rm R} = \prod_{\substack{i=1\\ \\  k_i \in \rm R}}^{2m} b^{\dagger}_{k_i}\ket{0}_{\rm R}\,,\quad n,m\in\mathbb{N}\cap[0,L/2]\,, 
\ee
In the case $-J<b < J$ we find 
\be
\ket{k_1, \ldots, k_{2n}}_{\rm NS} = \prod_{\substack{i=1\\ \\  k_i \in \rm NS}}^{2n} b^{\dagger}_{k_i}\ket{0}_{\rm NS}\,,\quad \ket{k_1, \ldots, k_{2m}}_{\rm R} = \prod_{\substack{i=1\\ \\  k_i \in \rm R}}^{2m} b^{\dagger}_{k_i}\ket{0}_{\rm R}\,,\quad n,m\in\mathbb{N}\cap[0,L/2]\,, 
\ee
In the case $J<b < -J$ we find 
\be
\ket{k_1, \ldots, k_{2n+1}}_{\rm NS} = \prod_{\substack{i=1\\ \\  k_i \in \rm NS}}^{2n+1} b^{\dagger}_{k_i}\ket{0}_{\rm NS}\,,\quad \ket{k_1, \ldots, k_{2m+1}}_{\rm R} = \prod_{\substack{i=1\\ \\  k_i \in \rm R}}^{2m+1} b^{\dagger}_{k_i}\ket{0}_{\rm R}\,, \quad n,m\in\mathbb{N}\cap[0,L/2]\,,
\ee
Finally for $b >|J|$ we find 
\be
\ket{k_1, \ldots, k_{2n}}_{\rm NS} = \prod_{\substack{i=1\\ \\  k_i \in \rm NS}}^{2n} b^{\dagger}_{k_i}\ket{0}_{\rm NS}\,,\quad \ket{k_1, \ldots, k_{2m+1}}_{\rm R} = \prod_{\substack{i=1\\ \\  k_i \in \rm R}}^{2m+1} b^{\dagger}_{k_i}\ket{0}_{\rm R}\,, \quad n,m\in\mathbb{N}\cap[0,L/2]\,.
\ee
For $L$ even we have 
\begin{align}
e^{i \pi N} \ket{0}_{\rm NS}&=\ket{0}_{\rm NS}\,, & e^{i \pi N} \ket{0}_{\rm R}&=\text{sign}(b+J)\text{sign}(b-J)\ket{0}_{\rm R}\,,\\
\Pi_L \ket{0}_{\rm NS}&=\ket{0}_{\rm NS}\,, & \Pi_L \ket{0}_{\rm R}&= (-1)^{\theta(-b-J)}\ket{0}_{\rm R}\,.
\end{align}
In the case $b <-|J|$ we find 
\be
\ket{k_1, \ldots, k_{2n}}_{\rm NS} = \prod_{\substack{i=1\\ \\  k_i \in \rm NS}}^{2n} b^{\dagger}_{k_i}\ket{0}_{\rm NS}\,,\quad \ket{k_1, \ldots, k_{2m+1}}_{\rm R} = \prod_{\substack{i=1\\ \\  k_i \in \rm R}}^{2m+1} b^{\dagger}_{k_i}\ket{0}_{\rm R}\,,\quad n,m\in\mathbb{N}\cap[0,L/2]\,, 
\ee
In the case $-J<b < J$ we find 
\be
\ket{k_1, \ldots, k_{2n}}_{\rm NS} = \prod_{\substack{i=1\\ \\  k_i \in \rm NS}}^{2n} b^{\dagger}_{k_i}\ket{0}_{\rm NS}\,,\quad \ket{k_1, \ldots, k_{2m}}_{\rm R} = \prod_{\substack{i=1\\ \\  k_i \in \rm R}}^{2m} b^{\dagger}_{k_i}\ket{0}_{\rm R}\,,\quad n,m\in\mathbb{N}\cap[0,L/2]\,, 
\ee
In the case $J<b < -J$ we find 
\be
\ket{k_1, \ldots, k_{2n}}_{\rm NS} = \prod_{\substack{i=1\\ \\  k_i \in \rm NS}}^{2n} b^{\dagger}_{k_i}\ket{0}_{\rm NS}\,,\quad \ket{k_1, \ldots, k_{2m}}_{\rm R} = \prod_{\substack{i=1\\ \\  k_i \in \rm R}}^{2m} b^{\dagger}_{k_i}\ket{0}_{\rm R}\,, \quad n,m\in\mathbb{N}\cap[0,L/2]\,,
\ee
Finally for $b >|J|$ we find 
\be
\ket{k_1, \ldots, k_{2n}}_{\rm NS} = \prod_{\substack{i=1\\ \\  k_i \in \rm NS}}^{2n} b^{\dagger}_{k_i}\ket{0}_{\rm NS}\,,\quad \ket{k_1, \ldots, k_{2m+1}}_{\rm R} = \prod_{\substack{i=1\\ \\  k_i \in \rm R}}^{2m+1} b^{\dagger}_{k_i}\ket{0}_{\rm R}\,, \quad n,m\in\mathbb{N}\cap[0,L/2]\,.
\ee

\subsection{Spectrum at the self dual points}
\label{app:integrablecasediff}
Let us consider the self dual points $J=s \pi/4$ and $b=r \pi/4$ where $r,s\in\{\pm 1\}$. For these points the dispersion relation drastically simplifies 
\be
\epsilon(k)= -\cos^{-1}(r s \cos(k))\!\!\!\!\mod 2\pi =\left(rs|\pi-k|-\pi \frac{(rs+1)}{2}\right)\!\!\!\!\mod 2\pi\,.
\label{eq:simpledispersion}
\ee
Writing the quantisation conditions explicitly we have 
\be
\epsilon(k_n)=\frac{2\pi}{L}\left(rs\left|\frac{L}{2}-n-\frac{\sigma}{2}\right|-\frac{L}{2} \frac{(rs+1)}{2}\right)\!\!\!\!\mod 2\pi\,,
\ee
where $\sigma=1$ in the NS sector and $\sigma=0$ in the R sector. This form implies 
\be
\epsilon(k_n)=\frac{2\pi}{L} \left(m+\frac{\sigma}{2}+ (L\!\!\!\!\!\mod2)  \frac{(rs-1)}{4}\right)\,, \qquad m\in\mathbb{N}\cap[0,L[\,.
\ee
In words, depending on the sector, on the parity of $L$, and on the sign of $rs$ the rescaled dispersion ${L}\epsilon(k_n)/{2\pi}$ is either an integer or a semi-integer number in $[0,L]$. Using this expression and the basis constructed above we can explicitly find the eigenvalues of the Floquet operator (3), it reads as 
\be
\varphi_{m,\rm r}= e^{- i e_{m,\rm r}}\,,\qquad\qquad{\rm r} \in\{{\rm NS}, {\rm R}\}\,,\qquad\qquad m\in\{1,\ldots,2^{L-1}\}\,.
\ee
Here we defined the quasi-energy
\be
e_{m,\rm r}=\frac{2\pi}{L} n_m + E_{\rm r} +\begin{cases}
0 & L\quad\text{even} \\
\frac{\pi (rs-1)}{4 L}  - \frac{\pi (rs-1)}{4 L } \eta_{\rm r}& L\quad\text{odd}\\
\end{cases}\,,\qquad n_m\in\mathbb{N}\cap[0,L[\,.
\ee
where we introduced the variable $\eta_{\rm r}$ such that $\eta_{\rm R}=-1$ and $\eta_{\rm NS}=1$. Moreover, we introduced the ground state energies $E_{\rm R},E_{\rm NS}$ in the two sectors, defined by 
\be
E_{\rm r}=-\frac{1}{2}\sum_{\substack{k\in {\rm r}}} \epsilon(k)\,.
\ee 
These expressions are explicitly evaluated using the form \eqref{eq:simpledispersion} of the dispersion relation, the result reads as 
\be
E_{\rm r}=\frac{\pi L}{4}+
\begin{cases}
0 & L\quad\text{even} \\
\frac{rs\pi}{4L} \eta_{\rm r}& L\quad\text{odd}\\
\end{cases}\,.
\ee
Putting all together we have 
\be
e_{m,\rm r}=\frac{2\pi}{L} n_m + \frac{\pi L}{4} +\begin{cases}
0 & L\quad\text{even} \\
\frac{\pi (rs-1)}{4 L}  + \frac{\pi}{4 L } \eta_{\rm r}& L\quad\text{odd}\\
\end{cases}\,,\qquad n_m\in\mathbb{N}\cap[0,L[\,.
\label{eq:explicit}
\ee
A direct consequence of \eqref{eq:explicit} is that quasi-energy differences are given by  
\be
e_{m,\rm r}-e_{m',\rm r'} = \frac{2\pi}{L} (n_m-n_{m'})+ \begin{cases}
0 & L\quad\text{even} \\
\frac{\pi}{4 L } (\eta_{\rm r}-\eta_{\rm r'})& L\quad\text{odd}\\
\end{cases}\,.
\label{eq:energydifference}
\ee

\section{Numerical methods}
\label{app:numerics}

Figure 2 is produced using direct time propagation of all basis states, namely we computed   
\be
\braket{{\boldsymbol{s}}|U_{\rm KI}^t |{\boldsymbol{s}}} = \braket{{\boldsymbol{s}}|(U_{\rm K}U_{\rm I})^t |{\boldsymbol{s}}}\,,
\label{eq:xx}
\ee
for each element $\ket{\boldsymbol{s}}$ of the ``computational basis'' composed of joint eigenstates of $\{\sigma^z_j\}$. $U_{\rm I}$ is diagonal in the computational basis, 
so its entries are stored in a $2^L$-sized register and used to multiply vectors repeatedly. The kick part $U_{\rm K}$ can be expressed as a simple tensor product of $2\times 2$ matrices
\be
 U_{\rm K}=\prod_j e^{i b \sigma^x_j}\equiv (\1\cos b + i \sigma^x \sin b)^{\otimes L},
\ee
where for $b=\pi/4$  
\begin{align}
U_{\rm K}=2^{-L/2} 
\begin{bmatrix} 
1& i \\
i & 1
\end{bmatrix}^{\otimes L} .
\end{align}
The matrix element (\ref{eq:xx}) can thus be computed in ${\cal O}(t L 2^L)$ operations, and the entire trace (spectral form factor) in ${\cal O}(t L 4^L)$ operations, which is advantageous to full diagonalization for $t\ll 2^L$. Using this algorithm we compute 
\be
K(t)=\sum_{\{s_j\in\{\pm1\}\}} \braket{{\boldsymbol s}|U_{\rm KI}^t[\boldsymbol{h}] |{\boldsymbol{s}}}\,,
\ee 
for specific realizations of disorder $\boldsymbol{h}=(h_1,\ldots,h_L)$, and then average over many realizations of disorder.

Table~I is produced by starting with a few random vectors of dimension $4^t$ and acting on them repeatedly with the transfer matrix $\mathbb T$, as in the ``power method''. The action of the transfer matrix can again be implemented efficiently as a multiplication by the eigenvalues of the diagonal part and kicks. After enough iterations the vectors are projected to the subspaces with the eigenvalues $\pm 1$. We extend this method to compute the gap shown in Fig. 3. We start with the random vector, make it orthogonal to $\pm 1$ sectors and look how the norm of this vector shrinks. After $n$ iterations, for large enough $n$, the norm changes by a factor of $|\lambda|^n$, where $\lambda$ is the second largest eigenvalue.


\begin{thebibliography}{99}
\bibitem{unfolding} Note that the unfolding procedure is not necessary for Floquet systems \cite{Haake} as the density of quasi-energies is expected to be uniform in the chaotic case.
\bibitem{note} By $U^*$ and $\ket{n}^*$ we designate complex conjugation in the canonical real basis, say eigenbasis $\ket{\bold{s}}$ of $\sigma^z_\tau$. We then have
	$\bra{m}^*\, U^* \,\ket{n}^* \equiv (\bra{m}U\ket{n})^*$.
\end{thebibliography}
\end{document}